\documentclass[journal]{IEEEtran}

\IEEEoverridecommandlockouts                              
\usepackage{cite}

\usepackage[utf8]{inputenc}
\usepackage{textcomp}
\usepackage{tikz,pgfplots}
\usepackage{graphicx}
\usepackage{mathrsfs} 
\usepackage{amsmath,amssymb,amsfonts,graphicx,float}

\usepackage{amsthm}
\usepackage{mathtools}
\usepackage[version=4]{mhchem}
\usepackage{siunitx}
\usepackage{longtable,tabularx}
\usepackage{algorithm}
\usepackage{algpseudocode}

\usepackage{xcolor}

\newtheorem{proposition}{Proposition}
\newtheorem{assumption}{Assumption}
\newtheorem{theorem}{Theorem}
\newtheorem{lemma}{Lemma}
\newtheorem{problem}{Problem}
\newtheorem{remark}{Remark}

\title{Safety-Constrained Learning and Control using Scarce Data and Reciprocal Barriers}

\author{Christos K. Verginis, Franck Djeumou, and Ufuk Topcu
\thanks{The authors are with the 
University of Texas at Austin, Austin, TX, USA. Email:
{\tt\small \{cverginis,fdjeumou, utopcu\}.utexas.edu}.
    }%
}

\begin{document}

\maketitle

\begin{abstract}
We develop a control algorithm that ensures the safety, in terms of confinement in a set, of a system with unknown, $2$nd-order nonlinear dynamics. The algorithm establishes novel connections between data-driven and robust, nonlinear control. It is based on data obtained online from the current trajectory and the concept of reciprocal barriers. More specifically, it first uses the obtained data to calculate set-valued functions that over-approximate the unknown dynamic terms. For the second step of the algorithm, we design a robust control scheme that uses these functions as well as  reciprocal barriers to render the system forward invariant with respect to the safe set. In addition, we provide an extension of the algorithm that tackles issues of controllability loss incurred by the nullspace of the control-direction matrix. The algorithm removes a series of standard, limiting assumptions considered in the related literature since it does not require global boundedness, growth conditions, or a priori approximations of the unknown dynamics’ terms.

\end{abstract}


\section{Introduction} \label{sec:intro}

Learning-based control is an important emerging topic of research that tackles uncertain autonomous systems. Control of uncertain systems has been widely studied in the literature, mostly by means of robust and adaptive control \cite{ioannou2012robust}. These techniques, however, require restrictive assumptions on the uncertainty type, such as linear parameterizations, a priori neural network approximations, or additive disturbances. 
Such assumptions might be too restrictive in cases where the dynamics sustain abrupt unknown changes, due to, for instance, unpredicted failures. Traditional control techniques might fail in such scenarios and one must turn to data-based approaches. At the same time, since we aim to tackle cases of abrupt dynamic changes, standard episodic reinforcement learning algorithms are inapplicable \cite{sutton2018reinforcement}; we are restricted to data obtained on the fly from the current trajectory, which limits greatly the available resources.

This paper  considers the problem of safety, in the sense of confinement in a given set, of $2$nd-order \textit{nonlinear systems} of the form (to be precisely defined in Sec.~\ref{sec:pf})
\begin{subequations} \label{eq:system 1}
\begin{align} 
\dot{x}_1 &= x_2 \\
\dot{x}_2 &= f(x) + g(x)u    
\end{align}
\end{subequations}
with \textit{a priori unknown} terms $f$ and $g$, for which the assumptions we impose are restricted to local Lipschitz continuity. 
Unlike previous works in the related literature, we do not impose growth conditions \cite{verginis2021adaptive} or \textit{global} Lipschitz continuity on the dynamics
\cite{franckACC,ornik2019control}, and we do not assume boundedness of the state \cite{ornik2019control}. Moreover, we do not restrict $g$ to be in the class of square positive definite matrices, a convenient property that has been commonly used in the related literature \cite{bechlioulis2008robust,verginis2019closed,liu2019barrier}. Finally, we do not 
employ a priori approximations of the system dynamics, such as linear parameterizations \cite{taylor2020adaptive,lopez2020robust} or neural networks \cite{bechlioulis2010prescribed}. 

Our proposed solution consists of a two-layered algorithm for the safety control of the unknown system in \eqref{eq:system 1}, integrating nonlinear feedback control with on the fly data-driven techniques. 
More specifically, the main contributions are as follows. 
Firstly, we use a discrete, finite set of data obtained from the current trajectory to compute an estimate of the control matrix $g$. 
Secondly, we use this estimate to design a novel feedback control protocol based on reciprocal barriers, rendering the system forward-invariant with respect to the given safe set under certain assumptions on the estimation error. We further provide an analytic relation between the estimation error and the frequency of the obtained data. Finally, we provide a provably correct extension that tackles controllability loss incurred by the control matrix $g$. 
The proposed algorithm is ``minimally invasive", in the sense that it acts only close to the boundary of the safe set, and does not require any expensive numerical operations or tedious analytic expressions to produce the control signal, enhancing thus its applicability.
\begin{figure}
	\centering
	\includegraphics[trim=.65cm 1cm 0 0,width=0.525\textwidth]{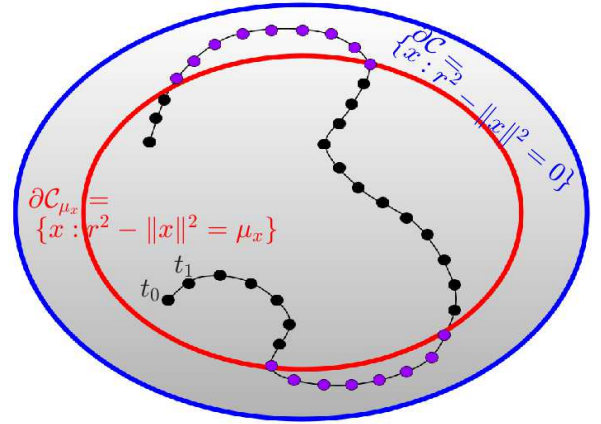}	
	\caption{ An example of a safe set $\mathcal{C} = \{x\in\mathbb{R}^n : h(x) = r^2 - \|x\|^2 > 0\}$ (with blue boundary) we aim to retain the system trajectory in, by using data obtained on the fly at the time instants $\{t_0,t_1,\dots\}$. The set $\mathcal{C}_{\mu_x} = \{x\in\mathbb{R}^n : h(x) = r^2 - \|x\|^2 \in (0,\mu_x)\}$ (with red boundary) dictates the region where the safety controller is activated.   }
	\label{fig:sets_1}
\end{figure}

This paper extends our preliminary work \cite{verginisFranck21} in the following directions: firstly, we consider second-order systems allowing position-like safety constraints with respect to $x_1$, in contrast to the first-order case of \cite{verginisFranck21}. Secondly, we provide better insight to the proposed solution by
relating the estimation error of $g$ with the frequency of the obtained data. Finally, we provide an analytic proof for the correctness of the algorithm that tackles the controllability loss incurred by the control matrix $g$. 

\textbf{Related Work:}
Safety of autonomous systems in the sense of set invariance \cite{blanchini1999set} is a topic that has been and still is undergoing intense study by the control community. The most widely used methodology is the concept of barrier certificates \cite{prajna2006barrier}, which provide a convenient and efficient way to guarantee invariance in a given set \cite{wieland2007constructive,ames2016control, ngo2005integrator}. 
Nevertheless, standard control based on barrier certificates relies heavily on the underlying dynamics since the respective terms are used in the control design. Extensions that tackle dynamic uncertainties have been considered in \cite{tee2009barrier,jankovic2018robust,xu2015robustness,taylor2020adaptive,lopez2020robust,jin2018adaptive,obeid2018barrier} using robust and adaptive control, restricted, however, to additive perturbations or linearly parameterized terms that include constant unknown parameters. Similarly, the recent works \cite{verginis2021sampling,verginis2019adaptive,verginis2019closed,arabi2020safety} use barrier functions to guarantee safety for systems whose dynamic terms satisfy linear parameterization with respect to uncertainties or growth and dissipative conditions. 
Therefore, the respective methodologies are not applicable to the class of systems considered in this paper.   

Another class of work dealing with unknown dynamics and using the concept of barrier certificates is that of funnel control, which guarantees confinement of the state in a given funnel,  \cite{bechlioulis2008robust,bechlioulis2010prescribed,hackl2013funnel,verginis2020asymptotic}. In contrast to the setup of the current paper, such methodologies either rely on approximation of the dynamics using neural networks  \cite{bechlioulis2010prescribed}, or require positive definite input matrix $g$  \cite{bechlioulis2008robust,verginis2020asymptotic} (see \eqref{eq:system 1}). The former 
has the drawbacks of lacking good heuristics for choosing radial basis functions and number of layers, as well as relying on strong assumptions on the amount of  data and the approximation errors. The latter is a convenient assumption on the controllability of the system; in fact, we show in this paper that when $g$ is square and positive definite, we achieve the safety of the system without resorting to the use of data (following similar steps with \cite{verginis2020asymptotic,bechlioulis2008robust}).


Barrier certificates have also been integrated with learning-based approaches to address the safety of uncertain systems \cite{cheng2019end,fisac2018general,jagtap2020control,liu2019barrier,srinivasan2020synthesis,robey2021learning,choi2020reinforcement}; \cite{cheng2019end,jagtap2020control,fisac2018general}, however, only consider additive uncertain terms modeled by Gaussian processes and assumed to evolve in compact sets. In \cite{choi2020reinforcement} the authors have access to a nominal model and propose an episodic reinforcement learning approach that tackles the residual disturbance. 
\cite{robey2021learning} and \cite{srinivasan2020synthesis} use data for learning barrier functions by employing the underlying dynamics, either partially or fully, and \cite{clark2019control} considers the safety problem for stochastic system with additive disturbances. Finally,
\cite{liu2019barrier} uses approximation of the dynamics using neural networks. On the contrary, we consider nonlinear systems of the form \eqref{eq:system 1} where both  $f$ and $g$ are unknown, without having access to any nominal model.

Moreover, many of the aforementioned works require large amounts of data in order to provide accurate resutls.
Recent methodologies that employ limited data obtained on the fly have been developed in \cite{ahmadi2020safe,ornik2019control,franckACC}, imposing, however, restrictive assumptions on the dynamics, such as global boundedness and Lipschitz continuity with known bounds, or known bounds on the approximation errors. In addition, the aforementioned works resort to online optimization techniques for safety specifications, increasing thus the complexity of the resulting algorithms. 
Other standard optimization-based algorithms that guarantee safety through state constraints \cite{frank_2003_nmpc_bible,herbert2019reachability} cannot tackle dynamic uncertainties more sophisticated than additive bounded disturbances. 
In the current paper, we rely on limited data without imposing any of the assumptions stated above, developing a computationally efficient safety control algorithm.


The remainder of this article is structured as follows. Section  \ref{sec:pf} gives the problem formulation. Section \ref{sec:dynamics approx} presents the approximation algorithm and the control design is provided in Section \ref{sec:control design}. Section \ref{sec:contr loss} investigates the case of controllability loss, and Section \ref{sec:simulations} presents simulation examples. Finally, Section \ref{sec:conclusion} concludes the paper.

\section{Problem Formulation} \label{sec:pf}

\subsection{Notation} 
We denote by $\bar{\mathbb{N}} \coloneqq \mathbb{N} \cup \{0\}$ the set of nonnegative integer numbers, where $\mathbb{N}$ is the set of natural numbers. The set of $n$-dimensional nonnegative reals, with $n\in\mathbb{N}$, is denoted by $\mathbb{R}^n_{\geq 0}$; $\textup{Int}(A)$, $\partial A$, and $\textup{Cl}(A)$ denote the interior,  boundary, and closure, respectively, of a set $A\subseteq \mathbb{R}^n$.
The open and closed ball of radius $r > 0$ around $x\in\mathbb{R}^n$ is denoted by ${\mathcal{B}}(x,r)$ and $\bar{\mathcal{B}}(x,r)$, respectively.
The minimum eigenvalue of a matrix $A\in \mathbb{R}^{n\times m}$ is denoted by $\lambda_{\min}(A)$. Given $a\in\mathbb{R}^n$, $\|a\|$ denotes its $2$-norm; $\nabla_y h(\cdot) \coloneqq \frac{\partial h(\cdot)}{\partial y}\in\mathbb{R}^m$ is the gradient of a real-valued function $h:\mathbb{R}^n\to\mathbb{R}$ with respect to $y\in \mathbb{R}^m$, and $\nabla h(x) \coloneqq \frac{\textup{d} h(x)}{\textup{d} x} \in \mathbb{R}^n$.  
An interval in $\mathbb{R}$ is denoted by $[a,b] = \{x\in\mathbb{R}:a \leq x \leq b\}$ and 
the set of intervals on $\mathbb{R}$ by $\mathbb{IR} \coloneqq \{ \mathcal{A} \coloneqq [\underline{\mathcal{A}},\bar{\mathcal{A}}] : \underline{\mathcal{A}},\bar{\mathcal{A}} \in \mathbb{R}, \underline{\mathcal{A}} \leq \bar{\mathcal{A}} \}$, which extends to the sets of interval vectors $\mathbb{IR}^n$ and matrices $\mathbb{IR}^{n\times m}$. We denote by $|\mathcal{A}| \coloneqq \max \{\overline{\mathcal{A}}, \underline{\mathcal{A}}\}$ the absolute value of an interval $\mathcal{A} \in \mathbb{IR}$,  and  the infinity norm of $\mathcal{B}  \coloneqq (\mathcal{B}_1,\hdots,\mathcal{B}_n) \in \mathbb{IR}^n$ by $\|\mathcal{B}\|_\infty = \max_{i\in \{1,\hdots,n\}} |\mathcal{B}_i|$. The width of an interval $\mathcal{A} \in \mathbb{R}^n$ is denoted by $\mathrm{wd}(\mathcal{A}) \coloneqq \overline{\mathcal{A}} - \underline{\mathcal{A}}$.
We carry forward the definitions \cite{moore1966interval} of arithmetic operations, set inclusion, and intersections of intervals to interval vectors and matrices 
componentwise.

\subsection{Problem Setup}

Consider a system characterized by $x \coloneqq [x_1,x_2]^\top \in\mathbb{R}^{2n}$, $x_i \coloneqq [x_{i_1},\dots,x_{i_n}]^\top \in \mathbb{R}^n$, $i\in\{1,2\}$, with dynamics 
\begin{subequations} \label{eq:system}
\begin{align} 
	\dot{x}_1 &= x_2, \\
	\dot{x}_2 &= f(x) + g(x)u
\end{align}
\end{subequations}
where $f\coloneqq [f_1,\dots,f_n]^\top:\mathbb{R}^{2n}\to \mathbb{R}^n$, $g\coloneqq [g_{ij}]:\mathbb{R}^{2n}\to\mathbb{R}^{n\times m}$ are \textit{unknown}, continuously differentiable functions, and $u \coloneqq [u_1,\dots,u_m]^\top \in\mathbb{R}^m$ is the control input. The problem this work considers is the invariance of the unknown system \eqref{eq:system} in a given closed set $\mathcal{C} \subset \mathbb{R}^n$ of the form 
\begin{align} \label{eq:C def}
	&\mathcal{C} \coloneqq \{ x_1 \in\mathbb{R}^n : h(x_1) \geq 0 \}  
\end{align}
where $h:\mathbb{R}^n \to \mathbb{R}$ is a continuously differentiable function, with bounded derivative $\frac{\textup{d} h(\mathsf{x})}{\textup{d} \mathsf{x}}$ in $\textup{Int}(\mathcal{C})$. 
More specifically, we aim to design a control law that achieves $x_1(t) \in \textup{Int}(\mathcal{C})$, i.e., $h(x_1(t)) > 0$, for all $t \geq t_0$, given that $x_1(t_0) \in \textup{Int}(\mathcal{C})$ for a positive   $t_0 \geq 0$. In our previous work \cite{verginisFranck21} we considered the safety of systems of the form $\dot{x} = f(x) + g(x)u$ in terms of $h(x) > 0$, assuming that $\nabla h(x)^\top g(x)$ is not identically zero, i.e., relative degree one. This, however, does not apply to safety specifications as dictated by \eqref{eq:C def} for systems of the form \eqref{eq:system}, whose relative degree is 2, since $\nabla h(x_1)^\top g(x) = 0$. Remark \ref{rem:high-order dyn} discusses the extension of the proposed algorithm to systems with higher relative degree.


As mentioned in Section \ref{sec:intro}, we aim to integrate a nonlinear feedback control scheme with a data-driven algorithm that approximates the dynamics \eqref{eq:system} by using data obtained on the fly from a finite-horizon trajectory. More specifically, consider an increasing time sequence $\{t_0, t_1, t_2, \dots \}$ signifying the time instants of data measurements.
That is, we assume that at each $t_i$, $i\in {\mathbb{N}}$, the system has access to the discrete dataset of $i$ points $\mathscr{T}_i \coloneqq \{(x^j,\dot{x}^j,u^j) \}^{i-1}_{j=0}$,  consisting of the system state $x^j = [x^j_1,\dots,x^j_n]^\top = x(t_j)$, the state derivative $\dot{x}^j = [\dot{x}^j_1,\dots,\dot{x}^j_n]^\top = \dot{x}(t_j)$, and the control input $u^j = [u^j_1,\dots,u^j_m]^\top = u(t_j)$ from a trajectory of \eqref{eq:system}. 
The  trajectory that produces the dataset $\mathscr{T}_i$ has finite horizon in the sense that, for each finite $i$, $\mathscr{T}_i$ is finite. 
We are now ready to give the problem statement treated in this paper. 

\begin{problem} \label{problem:1}
Let a system evolve subject to the unknown dynamics \eqref{eq:system}. Given the discrete dataset $\mathscr{T}_i$, $i\in {\mathbb{N}}$, and $x_1(t_0) \in \textup{Int}(\mathcal{C})$, compute a  timed-varying feedback control law $u:\mathbb{R}^{2n}\times [t_0,\infty) \to \mathbb{R}^m$ that guarantees $x_1(t) \in \textup{Int}(\mathcal{C})$, for all $t\geq t_0$.   
\end{problem}

We further impose the following assumptions, required for the solution of Problem \ref{problem:1}, where we use $\bar{\mathcal{A}} \coloneqq \mathcal{C} \times \mathbb{V}$, with $\mathbb{V} \subset \mathbb{R}^n$ a compact set.

\begin{assumption} \label{ass:C in ball}
 It holds that $\mathcal{C} \subset \mathcal{B}_r(0)$, where
 $\mathcal{B}_r(0)$ is the open ball of radius $r$ centered at $0$, for some $r>0$.     
\end{assumption}

\begin{assumption} \label{ass:lipshitz-bounds}
    There exist known positive constants $\bar{f}_k$, $\bar{g}_{k\ell}$ satisfying 
    $|f_k(\mathsf{x}) - f_k(\mathsf{y})| \leq \bar{f}_k |\mathsf{x} - 
    \mathsf{y}|$,  $|g_{k\ell}(\mathsf{x}) - g_{k\ell}(\mathsf{y})| \leq \bar{g}_{k\ell} |\mathsf{x} - \mathsf{y}|$, for all $k\in\{1,\dots,n\}$, $\ell\in\{1,\dots,m\}$, $\mathsf{x},\mathsf{y} \in \bar{\mathcal{A}}$.  
\end{assumption}

\begin{assumption} \label{ass:grad h}
    There exist positive constants $\nu_h$, $\varepsilon_h$ such that $\|\nabla h(x_1)\| \geq \varepsilon_h$ for all $x_1 \in \{ x_1 \in \mathbb{R}^n : h(x_1) \in (0,\nu_h] \} \subset \mathcal{C}$.
\end{assumption}


Assumption \ref{ass:C in ball} simply states that the system remains bounded when in the safe set $\mathcal{C}$. 
Assumption \ref{ass:lipshitz-bounds} considers knowledge of upper bounds of the Lipschitz constants of $f_k$ and $g_{kl}$ in the safe set $\bar{\mathcal{A}}$.
Note that we \textit{do not assume} that the system is bounded in any set or exact knowledge of the Lipschitz constants.  Assumption \ref{ass:grad h} is a simple controllability condition stating that the derivative of $h$ is not identically zero close to the boundary $\partial \mathcal{C}$. 

Note that the current problem setting exhibits a unique challenge due to the on the fly availability of the data measurements and the minor assumptions imposed on the dynamics \eqref{eq:system}.
In contrast to most related works, we do not assume global boundedness, Lipschitzness, or growth conditions on the dynamic terms, and we do not employ a priori approximation structures or data obtained offline.  

The solution of Problem \ref{problem:1}, consisting of a two-layered approach, is given in Sections \ref{sec:dynamics approx}-\ref{sec:contr loss}. Firstly, we use previous results on on the fly approximation of the unknown dynamics \cite{franckACC} and compute locally Lipschitz estimates for $g(x)$ (Section \ref{sec:dynamics approx}). Secondly, we use these estimates to design a closed-form feedback control law based on reciprocal barrier functions.

\section{ On-the-fly over-approximation of the dynamics} \label{sec:dynamics approx}

In this section, we provide a brief overview of the approximation algorithm of \cite{franckACC} based on data obtained online from a  single finite-horizon trajectory. More specifically, at each $t_i$, $i\in{\mathbb{N}}$, the algorithm uses the information from the finite dataset $\mathscr{T}_i$ in order to construct a data-driven differential inclusion $\dot{x} \in \mathcal{F}^i(x) + \mathcal{G}^i(x) u$ that contains the unknown vector fields of \eqref{eq:system}, where $\mathcal{F}^i : \mathbb{R}^{2n} \to \mathbb{IR}^n$ and $\mathcal{G}^i : \mathbb{R}^{2n} \to \mathbb{IR}^{n \times m}$ are known interval-valued functions. 
Such an over-approximation enables us to provide a locally Lipschitz  estimate $\hat{g}^i$ of $g$ to be used in the subsequent feedback control scheme. 



First, we propose in Lemma~\ref{lem:overapprox-f-G} closed-form expressions for $\mathcal{F}$ and $\mathcal{G}$ given over-approximations of $f$ and $g$ at some states. In the following, $\mathcal{A}$ is a closed subset of $\bar{\mathcal{A}}$.

\begin{lemma}[\textsc{\cite{franckACC}, Lemma $2$}]  \label{lem:overapprox-f-G}
	Let $i\in {\mathbb{N}}$ and consider the sets $\mathcal{A}$, $\mathscr{E}_i \coloneqq \{ (x^j,C^j_{\mathcal{F}}, C^j_{\mathcal{G}})\}_{j=0}^{i-1}$ where $C^j_{\mathcal{F}} \coloneqq (C^j_{\mathcal{F}_1},\dots,C^j_{\mathcal{F}_n})  \in\mathbb{IR}^n$, $C^j_{\mathcal{G}} \coloneqq (C^j_{\mathcal{G}_{k\ell}}) \in\mathbb{IR}^{n\times m}$ are intervals satisfying $f(x^j) \in C^j_{\mathcal{F}}$ and $g(x^j) \in C^j_{\mathcal{G}}$. Further, consider the locally Lipschitz constants $\bar{f}_k$, $\bar{g}_{k\ell}$ satisfying 
	$|f_k(\mathsf{x}) - f_k(\mathsf{y})| \leq \bar{f}_k |\mathsf{x} - \mathsf{y}|$,  $|g_{k\ell}(\mathsf{x}) - g_{k\ell}(\mathsf{y})| \leq \bar{g}_{k\ell} |\mathsf{x} - \mathsf{y}|$, for all $k\in\{1,\dots,n\}$, $\ell\in\{1,\dots,m\}$, $\mathsf{x},\mathsf{y} \in \mathcal{A}$ (see Assumption \ref{ass:lipshitz-bounds}). 
	The interval-valued functions $\boldsymbol{F}\coloneqq (\boldsymbol{F}_1,\dots,\boldsymbol{F}_n) : \mathbb{R}^{2n} \to \mathbb{IR}^{n}$ and $\boldsymbol{G} \coloneqq (\boldsymbol{G}_{k\ell}): \mathbb{R}^{2n} \to \mathbb{IR}^{n \times m}$, given for all $k\in\{1,\dots,n\}$ and $\ell \in \{1,\dots,m\}$, by the expressions
	\small
	\begin{subequations}	\label{eq:overapprox-f-G}
	\begin{align}
	\boldsymbol{F}_k(x) &\coloneqq  \bigcap_{(x^j, C^j_{\mathcal{F}}, \cdot) \in \mathscr{E}_i} \bigg\{ C^j_{\mathcal{F}_k} + \bar{f}_k \|x - x^j\| [-1,1] \bigg\}, \\
	\boldsymbol{G}_{k\ell}(x) &\coloneqq \bigcap_{(x^j, \cdot, C^j_{\mathcal{G}}) \in \mathscr{E}_i} \bigg\{ C^j_{\mathcal{G}_{k\ell}} + \bar{g}_{k\ell} \|x - x^j\| [-1,1] \bigg\},
	\end{align}
	\end{subequations}
	\normalsize
	satisfy $f(x) \in \boldsymbol{F}(x)$ 
	and $g(x) \in \boldsymbol{G}(x)$,
	for all $x \in \mathcal{A}$. 
\end{lemma}

Loosely speaking, Lemma~\ref{lem:overapprox-f-G} states that if a set $\mathscr{E}_i = \{ (x^j,C^j_{\mathcal{F}}, C^j_{\mathcal{G}}) \}_{j=0}^{i-1}$ and Lipschitz bounds are given, it is possible to obtain an analytic formula over the interval domain to over-approximate the unknown $f$ and $g$.  Lemma~\ref{lem:contraction} enables to compute the set $\mathscr{E}_i$ based on the dataset $\mathscr{T}_i$.
\begin{lemma}[{\cite{franckACC}, Lemma $1$}] \label{lem:contraction}
	Let a data point $(x^j,\dot{x}^j, u^j)$, a vector interval $\mathcal{F}^j \coloneqq (\mathcal{F}^j_1,\dots,\mathcal{F}^j_n) \in \mathbb{IR}^n$ such that $f(x^j) \in \mathcal{F}^j$, and a matrix interval $\mathcal{G}^j \coloneqq (\mathcal{G}^j_{k\ell}) \in \mathbb{IR}^{n \times m}$ such that $g(x^j) \in \mathcal{G}^j$. Consider the intervals $C^j_{\mathcal{F}} \coloneqq (C^j_{\mathcal{F}_1},\dots,C^j_{\mathcal{F}_n})\in\mathbb{IR}^n$ and $C^j_{\mathcal{G}} \coloneqq (C^j_{\mathcal{G}_{k\ell}})\in \mathbb{IR}^{n\times m}$, defined sequentially for $k\in \{1,\dots,n\}$, $\ell\in \{1,\dots,m\}$ by
	\small
    \begin{subequations} \label{eq:contraction-fG}
        \begin{align}
            C^j_{\mathcal{F}_k} & \coloneqq \left\{\mathcal{F}^j_k \right\}\: \cap \: \left\{\dot{x}^j_k - \mathcal{Y}^j_k\right\}, \\
            s_{0,k} & \coloneqq  \left\{ \dot{x}^j_k - C^j_{\mathcal{F}_k} \right\}\bigcap \{\mathcal{Y}^j_k\}, \\
            C^j_{\mathcal{G}_{k\ell}} & \coloneqq \begin{cases} \left( \left\{ s_{\ell-1,k} - \sum_{p>\ell} \mathcal{G}^j_{kp} u_p^j \right\} \bigcap  \left\{ \mathcal{G}^j_{k\ell} u_\ell^j \right\} \right) \frac{1}{u_\ell^j},
                                                 u_\ell^j \neq 0 \\
                                                \mathcal{G}^j_{k\ell}, \quad \hspace{35mm}\textup{otherwise},
                                    \end{cases}\\
            s_{\ell,k} &\coloneqq \left\{ s_{\ell-1,k} - C^j_{\mathcal{G}_{k\ell}} u_\ell^j \right\} \bigcap \left\{ \sum_{p > l} \mathcal{G}^j_{kp} u_p^j \right\}, 
        \end{align}
    \end{subequations}
    \normalsize
     where $\mathcal{Y}^j \coloneqq (\mathcal{Y}^j_1,\dots,\mathcal{Y}^j_n) \coloneqq \mathcal{G}^j u^j \in \mathbb{IR}^n$. Then, $C^i_\mathcal{F}$ and $C^i_\mathcal{G}$ 
    are the smallest intervals enclosing $f(x^j)$ and $g(x^j)$, respectively, given only the $(x^j,\dot{x}^j,u^j)$, $\mathcal{F}^j$, and $\mathcal{G}^j$.
\end{lemma}

\algdef{SE}[DOWHILE]{Do}{doWhile}{\algorithmicdo}[1]{\algorithmicwhile\ #1}%
\begin{algorithm}[!t]
	\caption{ $\mathsf{Approximate}(t_i,\mathscr{T}_i)$} 
	\label{algo:overapprox-datapoints}
	\begin{algorithmic}[1]    
		\Require{Single trajectory $\mathscr{T}_i$, sufficiently large $M >0$.
		}
		\Ensure{${\mathscr{E}_i = \{ (x^j,C^j_{\mathcal{F}}, C^j_{\mathcal{G}}) |  f(x^j) \in C^j_{\mathcal{F}}, g(x^j) \in C^j_{\mathcal{G}}\}_{j=0}^{i-1}}$}
		\State $\mathcal{A} \gets \bar{\mathcal{C}}$, $\mathcal{R}^{f_\mathcal{A}} \gets [-M,M]^n$, $\mathcal{R}^{G_\mathcal{A}} \gets [-M,M]^{n \times m}$ 
		\State Define $x^0 \in \mathcal{A}$, $C^0_{\mathcal{F}} \gets \mathcal{R}^{f_\mathcal{A}}$, and $C^0_{\mathcal{G}} \gets \mathcal{R}^{G_\mathcal{A}}$
		\For{$\iota \in \{1,\dots,i-1\} \wedge (x^\iota, \dot{x}^\iota, u^\iota) \in \mathscr{T}_{i}$} \label{alg:begin-init-e}
		\State Compute $\mathcal{F}^\iota \coloneqq \boldsymbol{F}(x^\iota), \mathcal{G}^\iota \coloneqq \boldsymbol{G}(x^\iota)$ via~\eqref{eq:overapprox-f-G} using $\mathscr{E}_{\iota-1}$		\label{alg:update-ei} 
		\State Compute $C^\iota_{\mathcal{F}}, C^\iota_{\mathcal{G}}$ via~\eqref{eq:contraction-fG}, $\mathcal{F}^\iota, \mathcal{G}^\iota$, and $(x^\iota,\dot{x}^\iota,u^\iota)$ \label{alg:contraction-f-G}
		\EndFor \label{alg:end-init-e}
		\Do \label{alg:while-begin}
		\State Execute lines~\ref{alg:begin-init-e}--\ref{alg:end-init-e} with $\mathscr{E}_i$ instead of $\mathscr{E}_{\iota-1}$ on line~\ref{alg:update-ei} 
		\doWhile{$\mathscr{E}_i$ is not invariant} \label{alg:while-end}
		\State \Return $\mathscr{E}_i$
	\end{algorithmic}
\end{algorithm}

Using Lemma~\ref{lem:contraction}, Alg.~\ref{algo:overapprox-datapoints} utilizes the dataset $\mathscr{T}_i$ at each time instant $t_i$, $i\in{\mathbb{N}}$, to compute the set $\mathscr{E}_i$ and subsequently, over-approximate $f$ and $g$ using~\eqref{eq:overapprox-f-G}. 
Note that the computational complexity of the algorithm  (in time and memory) is quadratic in the number of elements of $\mathscr{T}_i$ ( Lipschitz bounds estimations) and linear in the system dimension $n$. 
The subsequent theorem characterizes the correctness of the obtained differential inclusions. 

\begin{theorem}[{\cite{franckACC}, Theorem $1$}]\label{thm:diff incl}
Let $i \in \mathbb{N}$ and $\mathsf{F}^i \coloneqq (\mathsf{F}^i_1,\dots,\mathsf{F}^i_n):\mathbb{R}^{2n} \to \mathbb{IR}^n$, $\mathsf{G}^i \coloneqq (\mathsf{G}^i_{k\ell}):\mathbb{R}^{2n} \to \mathbb{IR}^{n\times m}$, with $\mathsf{F}^i(x) \coloneqq \boldsymbol{F}(x)$, $\mathsf{G}^i(x) \coloneqq \boldsymbol{G}(x)$ computed from \eqref{eq:overapprox-f-G} and the output $\mathscr{E}_i$ of Alg. \ref{algo:overapprox-datapoints}, which is executed at $t_i$ using the dataset $\mathscr{T}_i$.  Then it holds that $\dot{x}(t) \in \mathsf{F}^i(x(t)) + \mathsf{G}^i(x(t))u$, for all $t \geq t_i$. 
\end{theorem}

\begin{remark}
As pointed out in \cite{franckACC}, Alg. \ref{algo:overapprox-datapoints} can be adjusted to employ extra  information on $f$ and $g$, if available, yielding more accurate approximations. In particular, if we are given sets $\mathcal{A} \subseteq \bar{\mathcal{A}}$, $\mathcal{R}^{f_\mathcal{A}}$, $\mathcal{R}^{G_\mathcal{A}}$ such that $\{f(x) \vert x \in \mathcal{A}\} \subseteq \mathcal{R}^{f_\mathcal{A}}$ and $\{g(x) \vert x \in \mathcal{A}\} \subseteq \mathcal{R}^{G_\mathcal{A}}$, these can be used in Alg. \ref{algo:overapprox-datapoints}, replacing the respective ones defined in line 1. We stress, nevertheless, that such sets are not required to be available. 
\end{remark}


Based on Lemmas~\ref{lem:overapprox-f-G} and ~\ref{lem:contraction}, we propose now  a locally Lipschitz function 
$\hat{g}^i : \bar{\mathcal{A}} \to \mathbb{R}^{n \times m}$ that estimates the unknown function 
$g$ at each measurement instant $t_i$.

\begin{lemma}\label{lem:estimate}
    Let $i\in\mathbb{N}$. Given a weight $\theta \in [0,1]$ and a set $\mathcal{A} \subseteq \bar{\mathcal{A}}$, each component of the function 
    $\hat{g}^i \coloneqq [\hat{g}^i_{k\ell}]: \mathcal{A} \to \mathbb{R}^{n \times m}$, for all $x \in \mathcal{A}$, given by 
    \begin{align} \label{eq:estimate}
        \hat{g}^i_{k\ell}(x) &= \theta \underline{\mathsf{G}}^i_{k\ell}(x) + (1-\theta) \bar{\mathsf{G}}^i_{k\ell}(x),
    \end{align} 
    where $\underline{\mathsf{G}}^i_{k\ell}$ and $\bar{\mathsf{G}}^i_{k\ell}$ are the left and right endpoints, respectively, of the interval $\mathsf{G}^i_{k\ell}$, 
    is locally Lipschitz in $\mathcal{A}$, for all $k \in \{1,\hdots,n\}$ and $\ell \in \{1,\hdots,m\}$.
\end{lemma}
\begin{proof}
    By the definition of $\mathsf{G}^i_{kl}$, given by the output $\mathscr{E}_i$ of Alg.~\ref{algo:overapprox-datapoints} and Lemma~\ref{lem:overapprox-f-G}, we have for any $x \in \mathcal{A}$ that 
    \begin{align*}
        &\begin{aligned}
        \underline{\mathsf{G}}^i_{k\ell}(x) =&\inf \bigcap_{(x^j, \cdot, C^j_{\mathcal{G}_{k\ell}}) \in \mathscr{E}_i} \left\{ C^j_{\mathcal{G}_{k\ell}} + \bar{g}_{k\ell} ||x - x^j|| [-1,1] \right\},
        \end{aligned} \\
        &\begin{aligned}
            \bar{\mathsf{G}}^i_{k\ell}(x) = &\sup  \bigcap_{(x^j, \cdot, C^j_{\mathcal{G}_{k\ell}}) \in \mathscr{E}_i} \left\{ C^j_{\mathcal{G}_{k\ell}} + \bar{g}_{k\ell} ||x - x^j|| [-1,1], \right\}
        \end{aligned}   
    \end{align*}
    for all $k \in \{1,\hdots,n\}$ and $\ell \in \{1,\hdots,m\}$. Thus, one can  observe through interval arithmetic that $\underline{\mathcal{G}}^i_{k\ell}(x)$ and $\bar{\mathcal{G}}^i_{k\ell}(x)$ can be also written as
    \begin{align*}
         \underline{\mathsf{G}}^i_{k\ell}(x) &= \inf_{(x^j, \cdot, C^j_{\mathcal{G}_{k\ell}}) \in \mathscr{E}_i} \left\{m_\mathrm{u}(x^j) + \overline{g}_{k\ell} ||x - x^j||\right\}, \\
        \bar{\mathsf{G}}^i_{k\ell}(x) &= \sup_{(x^j, \cdot, C^j_{\mathcal{G}_{k\ell}}) \in \mathscr{E}_i} \left\{m_\mathrm{l}(x^j) - \bar{g}_{k\ell} ||x - x^j|| \right\},
    \end{align*}
    where $m_\mathrm{l}(x^j) \coloneqq \inf C^j_{\mathcal{G}_{k\ell}}$ and $m_\mathrm{u}(x^j) \coloneqq \sup C^j_{\mathcal{G}_{k\ell}}$. Using the inequality $|||\mathsf{x}|| - ||\mathsf{y}||| \leq ||\mathsf{x} - \mathsf{y}||$ for any two vectors $\mathsf{x},\mathsf{y}\in\mathbb{R}^n$,  we conclude that the functions $h^j : \mathsf{x} \mapsto m_\mathrm{u}(x^j) + \overline{g}_{k\ell} ||\mathsf{x} - x^j|| $ and $l^j : \mathsf{x} \mapsto m_\mathrm{l}(x^j) - \bar{g}_{k\ell} || \mathsf{x} - x^j||$ are Lipschitz continuous in $\mathcal{A}$ with $\bar{g}_{k\ell}$ as the Lipschitz constant in $\mathcal{A}$. Similarly, and by using $\max \{l,h\} = 0.5 (l+ h + |l-h|)$, and $\min \{l,h\} = 0.5 (l+h -|l-h|)$, we can deduce by induction that $ \mathsf{x} \mapsto \sup_{j \in \{1,\hdots,i\}} l_j(\mathsf{x})$ and $ \mathsf{x} \to \inf_{j \in \{1,\hdots,i\}} h_j(\mathsf{x})$ are also Lipschitz continuous owing to the Lipschitz continuity of $l^j$ and $h^j$ for all $j\in \{1,\hdots,n\}$. Thus, one obtains that $ x \to \bar{\mathsf{G}}^i_{k\ell}(x)$ and $x \mapsto \underline{\mathsf{G}}^i_{k\ell}(x)$ are also Lipschitz continuous, from which we conclude the Lipschitz continuity of $\hat{g}_{k\ell}$, being a linear combination of the latter. 
\end{proof}

\section{ Control design and Safety Guarantees} \label{sec:control design}

This section presents the main results of the paper. We first propose a learning-based control algorithm that relies on the approximation of the dynamics of Section \ref{sec:dynamics approx} and the concept of reciprocal barriers (section \ref{subsec:learning control}). Next, we provide in Section \ref{subsec: bounds} bounds on the approximation errors $\hat{g}_{kl}(x) - g_{kl}(x)$, $k\in\{1,\dots,n\}$, $\ell \in \{1,\dots,m\}$ based on the frequency of the update time instants $t_i$, $i\in \mathbb{N}$. Finally, Section \ref{subsec:square g} presents a simplified version of the algorithm for the special case where $g(x)$ is square and $g(x) + g(x)^\top$ is positive definite. 

\subsection{Learning-based Control Design} \label{subsec:learning control}

Given the set $\mathcal{C} = \{ x\in\mathbb{R}^n : h(x) \geq 0 \}$ and 
following \cite{ames2016control}, we define a continuously differentiable \textit{reciprocal} barrier function $\beta:(0,\infty) \to \mathbb{R}$ that satisfies \begin{align} \label{eq:beta cond 1}
	\frac{1}{\alpha_1(h)} \leq \beta(h) \leq \frac{1}{\alpha_2(h)} 
\end{align}
for class $\mathcal{K}$ functions $\alpha_1$, $\alpha_2$. Note that \eqref{eq:beta cond 1} implies $\inf_{x\in\textup{Int}(\mathcal{C})}\beta(h(x)) > 0$ and $\lim_{x \to \partial \mathcal{C}} \beta(h(x)) = \infty$. In order to render $\mathcal{C}$ forward-invariant, we aim to design a control algorithm that guarantees the boundedness of $\beta$ in a compact set. 
To this end, we enforce the extra condition on $\beta$: 
\begin{align} \label{eq:beta cond 3}
\max\left\{ \left\| \frac{\textup{d} \beta(h)}{\textup{d} h}\right\|,  \left\| \frac{\textup{d}^2 \beta(h)}{\textup{d} h^2}\right\| \right\} \leq \frac{1}{\alpha_4(h)},
\end{align} 
for a class $\mathcal{K}$ function $\alpha_4$, which essentially implies that the derivatives of  $\beta$ are bounded when $x_1$ lies in compact subsets of $\textup{Int}(\mathcal{C})$. 
Examples of $\beta$ include $\beta({h}) = \frac{1}{{h}}$, $\beta({h}) = - \ln\left( \frac{{h}}{1+{h}} \right)$. 

Ideally,  we would like the system to deviate minimally from a potential nominal task assigned to it, dictated by a nominal continuous control law $u_{\textup{n}}(x)$. Therefore, as in \cite{ames2016control}, we would like to allow $\beta$ to grow when $x_1$ is not close to the boundary of $\mathcal{C}$.  To this end, we use a switching signal $\sigma_\mu:\mathbb{R}\to\mathbb{R}_{\geq 0}$, defined, for a positive $\mu>0$,  as
\begin{align} \label{eq:sigma mu}
	\sigma_\mu(\mathsf{x}) \coloneqq \begin{cases}
	0, & \textup{if } \mathsf{x} \geq \mu \\
	\phi_\mu(\mathsf{x}), & \textup{if } \mathsf{x} \in [0,\mu) \\
	1, & \textup{if } \mathsf{x} \leq 0
	\end{cases}
\end{align}
where $\phi_\mu:[0,\infty) \to [0,1]$ is any \textit{decreasing} continuous function satisfying $\phi_\mu(0) = 1$ and $\phi_\mu(\mu) = 0$. For a given constant $\mu_x>0$, we want the control law to act only when $0 < h(x_1) \leq \mu_x$, which defines the set $\mathcal{C}_{\mu_x} \coloneqq \{ x_1 \in \mathbb{R}^{n}: h(x_1) \in (0,\mu_x]\}$. 
Therefore, following a backstepping-like scheme, we design first a reference signal for $x_2$ as 
\begin{equation} \label{eq:x_2des}
    x_{2,\textup{r}} \coloneqq x_{2,\textup{r}}(x_1) \coloneqq -  \kappa_x \sigma_{\mu_x}(h(x_1)) \beta_\textup{d}(x_1) \nabla h(x_1),
\end{equation}
where $\beta_\textup{d} \coloneqq \beta_\textup{d}(x_1)\coloneqq \frac{\textup{d} \beta(h(x_1))}{\textup{d} h(x_1)}$ and $\kappa_x$ is a constant control gain, and define the respective error 
\begin{equation} \label{eq:e_2}
    e_2 \coloneqq e_2(x) \coloneqq x_2 - x_{2,\textup{r}}
\end{equation}
We choose the constant $\mu_x$ chosen small enough such that $\mu_x < \nu_h$ implying $\|\nabla h(x_1)\| \geq \varepsilon_h$ for all $x_1 \in\mathcal{C}_{\mu_x}$ according to Assumption \ref{ass:grad h}.

We next design the control law such that $\beta$ and $e_2$ remain uniformly bounded.
To that end, we define a  
continuously differentiable function $h_v:\mathbb{R}^{2n} \to \mathbb{R}$ satisfying $h_v(x) > 0$ if and only if $\|e_2\| < \bar{B}_2$, for a positive $\bar{B}_2$, and define $\mathcal{C}_v \coloneqq \{ x \in \mathbb{R}^{2n}: h_v(x) \geq 0  \} \subset \mathbb{R}^{2n}$. An example is ellipsoidal functions of the form $h_v(x) \coloneqq 1 - e_2(x)^\top A_v e_2(x)$ for an appropriate  $A_v\in\mathbb{R}^{n\times n}$. The function $h_v$ depends explicitly on $e_2$, i.e., the difference $x_2 - x_{2,\textup{r}}(x_1)$, but we define it as a function of $x$ to ease the subsequent analysis.
We choose the constant $\kappa_x$ from \eqref{eq:x_2des} and $h_v$ such that $\mathcal{C}_v \subset \bar{\mathcal{A}}$ (see Assumption \ref{ass:lipshitz-bounds}), so that Algorithm \ref{algo:overapprox-datapoints} produces valid dynamic approximations in $\mathcal{A} \subset \bar{\mathcal{A}}$.
We also choose the constant $\bar{B}_2$ such that $h_v(x(t_0)) > 0$.
We next design the control law to guarantee  $h_v(x(t)) > 0$, for all $t \geq t_0$. Following similar steps as with $h(x_1)$, we define a continuously differentiable reciprocal function $\beta_v:(0,\infty)\to \mathbb{R}$ that satisfies 
\begin{align} \label{eq:beta cond 2}
	\frac{1}{\alpha_{v_1}(h_v)} \leq \beta_v(h_v) \leq \frac{1}{\alpha_{v_2}(h_v)} 
\end{align}
for class $\mathcal{K}$ functions $\alpha_{v_1}$, $\alpha_{v_2}$, and which we aim to maintain bounded. 
Similar to the definition of $\mathcal{C}_{\mu_x}$, we choose a constant $\mu_v > 0$, defining the set $\mathcal{C}_{v,\mu_v} \coloneqq \{ x \in \mathbb{R}^{2n}: h_v(x) \in (0,\mu_v] \}$  and aiming to enable the proposed control law only when $0 < h_v(x) \leq \mu_v$. 

Let now the estimate $\hat{g}^i(x)$ of $g(x)$, as computed by Lemma \ref{lem:estimate}, for $t\in [t_i,t_{i+1})$, $i\in\bar{\mathbb{N}}$\footnote{Since $\mathscr{T}_0$ is empty, $\hat{g}^0(x)$ is set randomly.}. 
Let also the respective error $\widetilde{g}^i(x) \coloneqq \hat{g}^i(x) - g(x)$ for $i\in \bar{\mathbb{N}}$. 
The term $\hat{g}^i(x)$ will be used in the control design to cancel the effect of $g(x)$, inducing thus a time-dependent switching. More specifically, 
the control law is designed as 
\begin{align} \label{eq:control law local}
u &\coloneqq u(x,t) \notag \\ 
 &\coloneqq u_\textup{n}(x) - \kappa_v \sigma_{\mu_v}(h_v(x)) \beta_{v,\textup{d}}(x)   \frac{\hat{g}^i(x)^\top \nabla h_v(x)}{ \|\hat{g}^i(x)^\top \nabla h_v(x)\|^2}, 
\end{align}
for $t\in[t_i,t_{i+1})$, $i\in\bar{\mathbb{N}}$, where $u_\textup{n}(x)$ is a nominal continuous controller, $\kappa_v$ is a positive constant control gain, and $\beta_{v,\textup{d}} \coloneqq \beta_{v,\textup{d}}(x)\coloneqq \frac{\textup{d} \beta_v(h_v(x))}{\textup{d} h_v(x)}$. 

The proof of correctness of \eqref{eq:control law local} is provided by Theorem \ref{th:local}, for which we need the next lemma, where we use $\bar{\mathcal{C}} \coloneqq \{ x\in\mathbb{R}^{2n}: x_1 \in \mathcal{C} \}$:
\begin{lemma} \label{lem:local solution}
Let a system evolve according to \eqref{eq:system}, \eqref{eq:x_2des}-\eqref{eq:control law local}, and a set $\mathcal{C}$ satisfying $x_1(t_0)\in \textup{Int}(\mathcal{C})$ for some $t_0 \geq 0$. 
Then there exists a unique, maximal solution $x:[t_0,t_{\max})\to \textup{Int}(\bar{\mathcal{C}})\cap \textup{Int}(\mathcal{C}_v)$ 
for some $t_{\max} > t_0$.
\end{lemma}
\begin{proof}
The closed-loop system $\dot{x} = f(x) + g(x)u(x,t)$ is piecewise continuous in $t \geq t_0$, for each fixed $x \in \textup{Int}(\bar{\mathcal{C}})\cap \textup{Int}(\mathcal{C}_v)$, and, in view of Lemma \ref{lem:estimate}, locally Lipschitz in $x \in \textup{Int}(\bar{\mathcal{C}})\cap \textup{Int}(\mathcal{C}_v)$ for each fixed $t \geq t_0$. Hence, since $\mathcal{C}_v$ is designed such that $x(t_0) \in \textup{Int}(\mathcal{C}_v)$, we conclude from \cite[Theorem 2.1.3]{bressan2007introduction} 
the existence of a unique, maximal, and absolutely continuous solution $x(t)$, satisfying $x(t) \in \textup{Int}(\bar{\mathcal{C}})\cap \textup{Int}(\mathcal{C}_v)$, for all $t\in[t_0,t_{\max})$, for a positive constant $t_{\max} > t_0$.
\end{proof}


Given the constant $\mu_v > 0$, we define now the set 
\begin{align*}
K_{\mu_v} \coloneqq \{& i \in \bar{\mathbb{N}} : \exists [\tau_1,\tau_2) \subseteq [t_0,t_{\max}), \textup{ with } \tau_1 \in [t_i, t_{i+1}) \\ 
&\textup{ s.t. } x(t) \in \textup{Int}(\bar{\mathcal{C}}) \cap \mathcal{C}_{v,\mu_v},\textup{ for all  } t\in[\tau_1,\tau_2) \}
\end{align*} 
where $t_i$ are the update instants defined in Section \ref{sec:dynamics approx}. The set $K_{\mu_v}$ contains the time index of the last update before entering the set $\mathcal{C}_{v,\mu_v}$ as well as the time indices of the updates while in the set $\mathcal{C}_{v,\mu_v}$ (like the purple points of the $x_1$-trajectory in Fig.~\ref{fig:sets_1}). Note that $K_{\mu_v}$ is not empty, {unless $x(t) \in \textup{Int}(\bar{\mathcal{C}})\cap \textup{Int}(\mathcal{C}_v) \backslash \mathcal{C}_{v,\mu_v}$ (i.e., $h(x_1(t)) > 0, h_v(x(t)) \geq \mu_v$ for all $t \geq t_0$).  }
The next theorem is the main result of this paper, stating that, if the estimated system defined by $\hat{g}^i(x)$ is controllable with respect to $h_v(x)$, and if $\hat{g}^i(x)$ is sufficiently close to $g(x)$, we achieve boundedness of $x_1(t)$ in $\textup{Int}(\mathcal{C})$ and, consequently, provide a solution to Problem \ref{problem:1}.

\begin{theorem} \label{th:local}
	Let a system evolve according to \eqref{eq:system}, \eqref{eq:x_2des}-\eqref{eq:control law local}, and a set $\mathcal{C}$ satisfying $x_1(t_0)\in \textup{Int}(\mathcal{C})$ for some $t_0 \geq 0$. 
	Let a constant $\mu_v' \in (0,\mu_v)$ and assume that there exists a positive constant $\varepsilon$ such
	that
	\begin{subequations}
	\begin{align}
	    & \|\hat{g}^i(x)^\top \nabla_{x_2} h_v(x)\| \geq \varepsilon, \label{eq:local cond1} \\    
	    & {\bar{g}_{\mu'_v}\bar{h}_{v,\mu_v'}}{ } <  \varepsilon \sigma_{\mu_v}(\mu_v'),
	    \label{eq:local cond2}
	\end{align}
	\end{subequations}
	for all $i \in K_{\mu_v'}$ and $x\in \mathcal{C}_{v,\mu_v'}$, where $\bar{h}_{v,\mu_v'} \coloneqq \sup_{x\in \mathcal{C}_{v,\mu_v'}}\|\nabla_{x_2} h_v(x)\|$ and $\bar{g}_{\mu_v'} \coloneqq \sup_{\substack{x\in \mathcal{C}_{v,\mu_v'}\\i\in K_{\mu_v'} }} \|\widetilde{g}^i(x)\|$.
	Then, under Assumptions \ref{ass:C in ball}, \ref{ass:grad h},   
	it holds that $x_1(t) \in \textup{Int}(\mathcal{C})$, and all closed loop signals are bounded, for all $ t \geq t_0$.
\end{theorem}
\begin{proof}	
    According to Lemma \ref{lem:local solution}, it holds that $x(t) \in \textup{Int}(\bar{\mathcal{C}}) \cap \textup{Int}(\mathcal{C}_v)$, for all $t\in[t_0,t_{\max})$ for a $t_{\max} > t_0$. 
    Assume now that $\lim_{t\to t_{\max}} h(x_1(t)) = 0$, i.e., the system converges to the boundary of $\mathcal{C}$ as $t\to t_{\max}$, implying $\lim_{t\to t_{\max}} \beta(h(x_1(t))) = \infty$. Let  $t_x' \in [t_0,t_{\max})$ and $\mu'_x\in(0,\mu_x)$ such that $x_1(t) \in \mathcal{C}_{\mu'_x} \subset \mathcal{C}_{\mu_x}$ for all $t\in[t_x',t_{\max})$, and 
	$x_1(t) \in \widetilde{\mathcal{C}}_x \coloneqq \{ x_1\in\mathbb{R}^n: h(x_1) \geq \min_{t\in[t_0,t']}\{h(x_1(t))\} > 0 \}$, for all $t\in[t_0,t_x']$\footnote{Note that such $t_x'$, $\mu_x'$ exist since $\lim_{t\to t_{\max}}h(x_1(t)) = 0$.}.
	Hence, it holds $0 < h(x_1(t)) \leq \mu_x' <
	\mu_x$ and $\sigma_{\mu_x}(h(x_1(t))) >  \sigma_{\mu_x}(\mu'_x) > 0$, for all $t\in[t_x',t_{\max})$. 
     In view of \eqref{eq:x_2des} and \eqref{eq:e_2}, $\dot{\beta}$ becomes
     \begin{align*}
        \dot{\beta} =&  \beta_{\textup{d}} \nabla h(x_1)^\top e_2 - \kappa_x \sigma_{\mu_x}(h(x_1)) \beta_{\textup{d}}^2 \|\nabla h(x_1) \|^2.
    \end{align*}
    In view of the definition of $h_v$ and since $x(t) \in \textup{Int}(\mathcal{C}_v)$ for all $t\in[t_0,t_{\max})$,  we conclude that $\|e_2(x(t)))\| < \bar{B}_2$, for all $t\in[t_0,t_{\max})$. Note that $\bar{B}_2$ is independent of $t_{\max}$. Moreover, since $\mu_x' < \mu_x < \nu_h$, Assumption \ref{ass:grad h} suggests that $\|\nabla h(x_1(t))\| \geq \varepsilon_h$ for all $t\in[t'_x,t_{\max})$. 
    Therefore,  $\dot{\beta}$ becomes
    \begin{align} \label{eq:beta_dot final}
        \dot{\beta} \leq & \bar{B}_2 | \beta_{\textup{d}}\| |\nabla h(x_1)\| -\kappa_x \sigma_{\mu_x}(h(x_1)) \beta_{\textup{d}}^2 \|\nabla h(x_1) \|^2 \notag \\
        \leq& \bar{B}_2 \bar{h}_x |\beta_\textup{d}| - \kappa_x \sigma_{\mu_x}(\mu'_x) \varepsilon_h^2\beta_\textup{d}^2  \notag \\
        =& -\kappa_x \sigma_{\mu_x}(\mu'_x) \varepsilon_h^2 |\beta_\textup{d}|\left( |\beta_\textup{d}| - \frac{\bar{B}_2 \bar{h}_x}{\kappa_v \sigma_{\mu_x}(\mu'_x) \varepsilon_h^2} \right)
    \end{align}
    for all $t\in[t_x',t_{\max})$, where $\bar{h}_x \coloneqq \sup_{x_1\in\mathcal{C}_{\mu'_x}}\|\nabla h(x_1)\|$ is a finite constant, since $h(x_1)$ is continuously differentiable. Therefore, we conclude that $\dot{\beta} < 0$ when $|\beta_\textup{d}| > \frac{\bar{B}_2\bar{h}_x}{\kappa_x \sigma_{\mu_x}(\mu'_x) \varepsilon_h^2}$. 
    
    We claim now that \eqref{eq:beta_dot final} implies the boundedness of $\beta$. Since we have assumed that $\lim_{t\to t_{\max}} \beta (h(x_1(t))) = \infty$, \eqref{eq:beta cond 1} and \eqref{eq:beta cond 3} imply that $\lim_{t\to t_{\max}}|\beta_\textup{d}(t)| = \infty$. Hence, for every positive constant $\gamma > 0$, there exists a time instant $t_\gamma \in [t_x',t_{\max} )$ such that $|\beta_\textup{d}(t)| > \gamma$ for all $t > t_\gamma$. Consequently, we conclude from \eqref{eq:beta_dot final} that there exists a time instant $t' \in [t_x',t_{\max})$ such that ${\beta}(h(x_1(t))) < 0$ for all $t > t'$, which leads to a contradiction.  
	We conclude, therefore, that there exists a constant $\bar{\beta}$ such that $\beta(h(x_1(t)) \leq \bar{\beta}$, for all $t\in[t_0,t_{\max})$, implying $h(x_1(t)) \geq \underline{h} \coloneqq \alpha_1^{-1}\left(\frac{1}{\bar{\beta}}\right)$, for all $t\in[t_0,t_{\max})$, which dictates  the boundedness of $x_1$ in a compact set $x_1(t) \in \widetilde{\mathcal{C}} \subset \textup{Int}(\mathcal{C})$, for all $t\in[t_0,t_{\max})$. 
	Moreover, \eqref{eq:x_2des} also suggests the boundedness of $x_{2,\textup{r}}(x_1(t))$, which, via the boundedness of $e_2$ by $\bar{B}_2$, implies the boundedness of $x_2(t)$, for all $t\in[t_0,t_{\max})$. By differentiating \eqref{eq:x_2des} and using the boundedness of $x_1$, $x_{2,\textup{r}}$, and \eqref{eq:beta cond 3}, \eqref{eq:e_2}, we also conclude the boundedness of $\dot{x}_{2,\textup{r}}(x_1(t)))$, for all $t\in[t_0,t_{\max})$.
	
	We proceed next to prove the boundedness of $\beta_v$. Following the same line of proof, assume that $\lim_{t\to t_{\max}} h_v(x(t)) = 0$, i.e., the system converges to the boundary of $\mathcal{C}_v$ as $t\to t_{\max}$, implying $\lim_{t\to t_{\max}} \beta_v(h_v(x(t))) = \infty$. Given the constant $\mu_v'\in(0,\mu_v)$, let any $t_v' \in [t_0,t_{\max})$  such that $x(t) \in \mathcal{C}_{v,\mu'_v} \subset \mathcal{C}_{v,\mu_v}$ for all $t\in[t_v',t_{\max})$, and 
	$x(t) \in \widetilde{\mathcal{C}}_v \coloneqq \{ x\in\mathbb{R}^{2n}: h_v(x) \geq \min_{t\in[t_0,t']}\{h_v(x(t))\} > 0 \}$, for all $t\in[t_0,t_v']$.
	Hence, it holds $0 < h_v(x(t)) \leq \mu_v' <
	\mu_v$ and $\sigma_{\mu_v}(h_v(x(t))) >  \sigma_{\mu_v}(\mu'_v) > 0$, for all $t\in[t_v',t_{\max})$.  
	By recalling that $\sigma_{\mu_v}(h_v) \leq 1$ and that $h_v(\cdot)$ is a function of $e_2$, the derivative of $\dot{\beta}_v$ becomes then
	\begin{align*}
	    \dot{\beta}_v =& \beta_{v,\textup{d}} f_\textup{n}(x) \\ 
	    &\hspace{-3mm} - \kappa_v \sigma_{\mu_v}(h_v) \beta_{v,\textup{d}} \nabla_{x_2} h_v(x)^\top g(x) \beta_{v,\textup{d}}\frac{\hat{g}^i(x)^\top \nabla_{x_2} h_v(x)}{ \|\hat{g}^i(x)^\top \nabla_{x_2} h_v(x)\|^2} \\
	    \leq &\beta_{v,\textup{d}} f_{\textup{n}}(x)  \\
	    & -\kappa_v \sigma_{\mu_v}(h_v) \beta_{v,\textup{d}}^2 + \kappa_v \beta_{v,\textup{d}}^2 \frac{\widetilde{g}^i(x)^\top \nabla_{x_2} h_v(x)}{ \|\hat{g}^i(x)^\top \nabla_{x_2} h_v(x)\|}
	\end{align*}
	for all $t\in[t_v',t_{\max})$, where $f_\textup{n}(x) \coloneqq \nabla_{x_1} h_v(x)^\top x_2 + \nabla_{x_2} h_v(x)( f(x) + g(x)u_\textup{n}(x))$ and we used $g(x) = \hat{g}^i(x) - \widetilde{g}^i(x)$.  
	Since $f,g,u_\textup{n}$ are continuous functions, $\dot{x}_{2,\textup{r}}$ has been proven bounded, and $x(t)\in\textup{Int}(\bar{\mathcal{C}})\cap\textup{Int}(\mathcal{C}_v)$ for all $t\in[t_0,t_{\max})$, there exists a constant $\bar{f}_\textup{n}$, independent of $t_{\max}$, satisfying $|f_\textup{n}(x(t)) | \leq \bar{f}_\textup{n}$, for all $t\in[t_0,t_{\max})$.  By also using $\|\hat{g}^i(x(t))^\top \nabla_{x_2} h_v(x(t))\| \geq \varepsilon$, for all $t\in[t',t_{\max})$, we obtain
	\begin{align*}
	\dot{\beta} \leq |\beta_{v,\textup{d}}| \bar{f}_\textup{n} - \kappa_v \sigma_{\mu_v}(\mu'_v)\beta_{v,\textup{d}}^2 + \kappa_v \beta_{v,\textup{d}}^2 \frac{\bar{h}_{v,\mu_v'} \bar{g}_{\mu'_v}}{\varepsilon},
	\end{align*} 
	for all $t\in[t',t_{\max})$.
	By setting $\epsilon_{v} \coloneqq \sigma_{\mu_v}(\mu'_v) - \frac{\bar{h}_{v,\mu'_v}\bar{g}_{\mu'_v}}{\varepsilon}>0$, we obtain
	\begin{align*}
	\dot{\beta}_v &\leq  - \kappa_v \epsilon_{v} |\beta_{v,\textup{d}}| \left( |\beta_{v,\textup{d}}|  - \frac{\bar{f}_\textup{n}}{\kappa_v \epsilon_v} \right) 
	\end{align*}
	for all $t\in[t',t_{\max})$. Therefore, $\dot{\beta}_v < 0$ when $|\beta_{v,\textup{d}}| > \frac{ \bar{f}_\textup{n}}{\kappa_v \epsilon_v}$. 
	By invoking similar arguments as in the case of  \eqref{eq:beta_dot final}  and $\beta$, we conclude that $x(t) \in \widetilde{\mathcal{C}} \subset \textup{Int}(\bar{\mathcal{C}})\cap \textup{Int}(\mathcal{C}_v)$.
	
	By also using $x(t) \in \widetilde{\mathcal{C}}_v$, for all $t\in[t_0,t']$ and the compactness of $\widetilde{\mathcal{C}}_v$, we conclude the boundedness of $x(t)$ and $\beta_v(h_v(x(t)))$, for all $t\in[t_0,t_{\max})$.  From  \cite[Th. 2.1.4]{bressan2007introduction}, we conclude that $t_{\max} = \infty$, and the boundedness of $x(t)$, $\beta(h(x_1(t)))$, $u(x(t),t)$ for all $t\in[t_0,\infty)$.
\end{proof}

Intuitively, the condition ${\bar{g}_{\mu'_v}\bar{h}_{v,\mu'_v}} < {\varepsilon}  \sigma_{\mu}(\mu_v')$ of Th. \ref{th:local} is implicitly connected to the frequency of measurements $\{x^i,\dot{x}^i,u^i\}$ the system obtains. More specifically, note first that $\|\widetilde{g}^i(x(t_i)) \| \leq \lim_{t \to t_i^{-}} \|\widetilde{g}^i(x(t_i)) \|$, for all $i\in \bar{\mathbb{N}}$, i.e., the estimation $g^i(x)$ of $g(x)$ improves with every update. Hence, the condition ${\bar{g}_{\mu'_v}\bar{h}_{v,\mu'_v}} < {\varepsilon}  \sigma_{\mu}(\mu_v')$ simply implies that the system will have obtained sufficiently enough measurements such that it obtains an accurate enough estimate $\hat{g}^i(x)$ of $g(x)$ from Alg. \ref{algo:overapprox-datapoints} before it reaches $\partial \mathcal{C}$ or $\partial\mathcal{C}_v$. 

\subsection{Bounds on $\widetilde{g}^i(x)$} \label{subsec: bounds}

Theorem \ref{th:local} is based on a small enough error of the approximation error $\widetilde{g}^i(x)$. Intuitively, this is achieved through a sufficiently high frequency of the measurement updates $t_i$, $i\in\mathbb{N}$. In this section we provide a closed-form relation between the approximation error $\widetilde{g}^i(x^{i+1})$ and the difference update $\Delta t_i \coloneqq t_{i+1} - t_i$. 

We start by stating a result from~\cite{djeumou2020onthefly}, which provides a closed-form expression to over-approximate future state values of the system given the current state and the control signal applied. Since Theorem \ref{th:local} proves the boundedness of the control signal $u(x,t)$, we use $\mathcal{U}^i \in \mathbb{IR}^m$ to denote the bounded set satisfying $u(x(t),t) \in \mathcal{U}^i \in \mathbb{IR}^m$, for all $t\in[t_i,t_{i+1}]$ for $i\in\bar{\mathbb{N}}$.
\begin{theorem}[\cite{djeumou2020onthefly}, Theorem $2$]\label{thm:next-state}
    Let the current state be $x^i$, the bounded admissible set of control values between time $t_i$ and $t_{i+1} = t_{i}+\Delta t_i$ be $\mathcal{U}^i \in \mathbb{IR}^m$, with a time step size $\Delta t_i > 0$. Assume that $ (\sqrt{n} \beta^{i})\Delta t < 1 $, where $\beta^{i} \coloneqq \sqrt{ \sum_{k=1}^n (\overline{f}_k + \sum_{l=1}^m \overline{g}_{k\ell} |\mathcal{U}^i_l|})^2$, and 
    $\bar{f}_k$, $\bar{g}_{k\ell}$ are the known locally Lipschitz constants (see Assumption \ref{ass:lipshitz-bounds}). Then, the future state value $x^{i+1} = x(t_{i+1})$ satisfies
    \begin{align}\label{eq:next-state-pred}
        x^{i+1} \in x^i + \boldsymbol{h}(x^{i}, \mathcal{U}^i) \Delta t + ( \mathcal{J}^f + \mathcal{J}^g \mathcal{U}^i) \boldsymbol{h}(\mathcal{S}^{i}, \mathcal{U}^i)\frac{\Delta t_i^2}{2},
    \end{align}
    where $\boldsymbol{h}(x,u) \coloneqq \boldsymbol{F}(x) + \boldsymbol{G}(x) u$, $\mathcal{J}^f\coloneqq (\mathcal{J}^{f}_{kp}) \in \mathbb{IR}^{n \times n}$ and $\mathcal{J}^g \coloneqq (\mathcal{J}^g_{k\ell p}) \in \mathbb{IR}^{n \times m \times n}$ are over-approximations of the Jacobian of $f$ and $g$ given by $\mathcal{J}^{f}_{kp} = \bar{f}_k [-1,1]$ and $\mathcal{J}^g_{k\ell p} = \bar{g}_{k\ell} [-1,1]$ for all $k,p \in \{1,\dots,n\}$ and $\ell \in \{1,\dots,m\}$. Moreover, $\mathcal{S}^i$ is an a priori rough enclosure of the future state value given by
    \begin{align}
        \mathcal{S}^i \coloneqq x^i +  \frac{\Delta t \|\boldsymbol{h}(x^i,\mathcal{U}^i)\|_{\infty}}{1- \sqrt{n} \Delta t_i \beta^i} [-1,1]^n. \label{eq:explict-fixpoint}
    \end{align}
\end{theorem}
Note that, with a slight abuse of notation, 
$\boldsymbol{F}$ and $\boldsymbol{G}$ is re-defined in Theorem~\ref{thm:next-state} to take both real vectors and interval quantities as arguments,  
as expressed by $\boldsymbol{h}(x^{i}, \mathcal{U}^i)$ and 
$\boldsymbol{h}(\mathcal{S}^{i}, \mathcal{U}^i)$. To achieve this, one can straightforwardly extend the $\|\cdot\|$ operator to the domain of intervals~\cite{franckACC}. In the remainder of this section, $\|\cdot\|$ takes both real and interval vectors as arguments.

Finally, based on the estimate~\eqref{eq:estimate} in Lemma~\ref{lem:estimate} and Theorem~\ref{thm:next-state}, we provide in Lemma~\ref{lem:error-estimate} an upper bound on the error between the unknown function $g$ and its estimate $\hat{g}^i$.

\begin{lemma}[\textsc{Point-based estimation error}]\label{lem:error-estimate}
    Under the notation and assumption of Theorem~\ref{thm:next-state}, it holds that
    \begin{align}\label{eq:estimate-error}
        \| \hat{g}^i_{k\ell`}(x^{i+1}) - g_{k\ell}(x^{i+1}) \| &\leq \mathrm{wd}(C^i_{\mathcal{G}_{k\ell}}) \nonumber + 2 \bar{g}_{k\ell}  \overline{\|h(x^i,\mathcal{U}^i)\|} \Delta t_i \nonumber \\
        &\quad + 2 \bar{g}_{k\ell} \overline{\|\mathcal{K}^i\|}\frac{\Delta t_i^2}{2}
    \end{align}     
    for $k \in \{1,\hdots,n\}$, and $l \in \{1,\hdots,m\}$, where 
    \begin{align*}
        \mathcal{K}^i \coloneqq (\mathcal{J}^f + \mathcal{J}^g\mathcal{U}^i) \Big(h(x^i,\mathcal{U}^i) + \frac{\Delta t_i \|\boldsymbol{h}(x^i,\mathcal{U}^i)\|_{\infty}}{1- \sqrt{n} \Delta t_i \beta^i} \mathcal{H}^i \Big)  \in \mathbb{IR}^n,
    \end{align*} 
    and $\mathcal{H}^i \coloneqq (\bar{f} + \bar{g} \mathcal{U})[-\sqrt{n},\sqrt{n}]^n \in \mathbb{IR}^n$.
\end{lemma}

\begin{proof}
    By definition of $\hat{g}_{k\ell}(x) \in \boldsymbol{G}_{k\ell}(x)$, we have
    \begin{align}\label{eq:error-fv}
        \| \hat{g}_{k\ell}(x^{i+1}) - g_{k\ell}(x^{i+1})\| \leq \mathrm{wd}(\boldsymbol{G}_{k\ell}(x^{i+1})),
    \end{align}
    since we know by construction that $g_{k\ell}(x^{i+1}) \in  \boldsymbol{G}_{k\ell}(x^{i+1})$ for all $k\in\{1,\dots,n\}$ and $\ell\in\{1,\dots,m\}$. As a consequence, by construction of $\boldsymbol{G}_{k\ell}$ in Theorem~\ref{thm:diff incl} and Lemma~\ref{lem:overapprox-f-G}, we can deduce that
    \begin{align}\label{eq:width-next}
        \mathrm{wd}(\boldsymbol{G}_{k\ell}(x^{i+1})) \leq \mathrm{wd}(C^i_{\mathcal{G}_{k\ell}} + \bar{g}_{k\ell} \|x^{i+1}-x^i\|)] [-1,1]).
    \end{align}
    We use the relation~\eqref{eq:next-state-pred} of Theorem~\ref{thm:next-state} to bound the quantity $\|x^{i+1}-x^i\|$ as
    \begin{align}\label{eq:diff-overapprox}
        \|x^{i+1} - x^{i}\| &\in \left\|\boldsymbol{h}(x^i, \mathcal{U}^i) \Delta t_i + (\mathcal{J}^f + \mathcal{J}^g \mathcal{U}^i) \boldsymbol{h}(\mathcal{S}^i, \mathcal{U}^i)\frac{\Delta t_i^2}{2}\right\|.
    \end{align}
    Then, using the definition of $\mathcal{S}^i$ from~\eqref{eq:explict-fixpoint}, we have that 
    \begin{align}
        \boldsymbol{F}_k (\mathcal{S}^i) &\subseteq C^i_{\mathcal{F}_k} + \bar{f}_k \frac{\Delta t_i \|\boldsymbol{h}(x^i,\mathcal{U}^i)\|_{\infty}}{1- \sqrt{n} \Delta t_i \beta^i} [-\sqrt{n},\sqrt{n}],\\
        \boldsymbol{G}_{k,l} (\mathcal{S}^i) &\subseteq C^i_{\mathcal{G}_{k\ell}} +  \bar{g}_{k\ell} \frac{\Delta t_i \|\boldsymbol{h}(x^i,\mathcal{U}^i)\|_{\infty}}{1- \sqrt{n} \Delta t_i \beta^i} [-\sqrt{n},\sqrt{n}].
    \end{align}
    Hence, we can write
    \begin{align} \label{eq:h-aprior-val}
        \boldsymbol{h}(\mathcal{S}^i, \mathcal{U}^i) \subseteq \boldsymbol{h}(x^i,\mathcal{U}^i) + \frac{\Delta t_i \|\boldsymbol{h}(x^i,\mathcal{U}^i)\|_{\infty}}{1- \sqrt{n} \Delta t_i \beta^i} \mathcal{H}^i,
    \end{align} 
    Finally, merging~\eqref{eq:h-aprior-val} into~\eqref{eq:diff-overapprox}, and plugging the result into~\eqref{eq:width-next} and then~\eqref{eq:error-fv} enables to obtain the error~\eqref{eq:estimate-error}.
\end{proof} 
Relation \eqref{eq:estimate-error} provides a way to bound the error $\widetilde{g}^i(x)$ by using a small enough time difference $\Delta t_i$.
Note that, apart from the Lipschitz constant estimates of Assumption \ref{ass:lipshitz-bounds}, \eqref{eq:estimate-error} does not require any additional information on the dynamic terms $f(\cdot)$ and $g(\cdot)$.
Therefore, the necessary condition \eqref{eq:local cond2} of Theorem \ref{th:local} could be potentially replaced by a specification for $\Delta t_i$ though  \eqref{eq:estimate-error}.

\subsection{ Square Control Matrix $g(x)$} \label{subsec:square g}

The fact that $g$ is non-square and completely unknown and hence cannot be used in the control algorithm makes the considered problem significantly more challenging compared to other works in the related literature that assume positive definiteness of $g$ or of ${g + g^\top}$ \cite{bechlioulis2008robust,liu2019barrier,verginis2020asymptotic,verginis2019closed}. In fact, we show now that a simple feedback control law can solve Problem \ref{problem:1} in the case of square and positive definite matrix $g$, without using any data. 
We first need Assumption \ref{ass:grad h} to hold for $\nabla_{x_2} h_2$:
\begin{assumption} \label{ass:grad h2}
    There exist positive constants $\nu_v$, $\varepsilon_v$ such that $\|\nabla_{x_2} h_v(x)\| \geq \varepsilon_v$ for all $x \in \{ x \in \mathbb{R}^{2n} : h_v(x) \in (0,\nu_v] \} \subset \mathcal{C}_v$.
\end{assumption}
As with \eqref{eq:control law local}, we select a positive constant $\mu_v$ to enable the control law in the set $\mathcal{C}_{v,\mu_v} =\{x\in\mathbb{R}^{2n} : h_v(x) \in (0,\mu_v]\}$. Similarly to $\mu_x$, we choose the constant $\mu_v$ sufficiently small so that it satisfies $\mu_v < \nu_v$, implying $\|\nabla_{x_2} h_v(x)\| \geq \varepsilon_v$ for all $x\in \mathcal{C}_{\mu_v}$. 

Given the reference signal $x_{2,\textup{r}}$ in \eqref{eq:x_2des}, the function $h_v(\cdot)$ and $\beta_{v}(\cdot)$ in \eqref{eq:beta cond 2},  the switching function \eqref{eq:sigma mu}, and the constant $\mu_v$,  we design now the control law as 
\begin{align} \label{eq:control law square}
u \coloneqq u(x,t) &\coloneqq u_\textup{n}(x) - \kappa_v \sigma_{\mu_v}(h_v(x)) \beta_{v,\textup{d}}(x) \nabla_{x_2} h_v(x),
\end{align}
whose correctness is proven in the following theorem.

\begin{theorem} \label{th:square}
Let a system evolve according to \eqref{eq:system 1}, \eqref{eq:control law square}, with $m=n$ and a set $\mathcal{C}$ satisfying $x_1(t_0) \in \textup{Int}(\mathcal{C})$ for some $t_0 \geq 0$. Let a constant $\mu'_v \in (0,\mu)$ and assume that 
$\lambda_{\min}\left({g(x) + g(x)^\top}\right) > 0$, for all $x\in  \mathcal{C}_{v,\mu'_v}$. Under Assumptions \ref{ass:C in ball}, \ref{ass:grad h}, it holds that $x_1(t) \in \textup{Int}(\mathcal{C})$, and all closed loop signals are bounded, for all $ t \geq t_0$.
\end{theorem}
\begin{proof}
The proof follows similar steps as in the proof of Theorem \ref{th:local} and only a sketch is given. Firstly, we establish a unique, continuously differentiable, and maximal solution $x:[t_0,t_{\max}) \to \textup{Int}(\bar{\mathcal{C}}) \cap \textup{Int}(\mathcal{C}_v)$, for some $t_{\max} > t_0$. By differentiating $\beta$, we obtain \eqref{eq:beta_dot final}, which guarantees  the boundedness of $\beta$ as $\beta(h(x_1(t)) \leq \bar{\beta}$ and the boundedness of $x_1$, $x_2$, $x_{2,\textup{r}}$, and $\dot{x}_{2,\textup{r}}$ for all $t\in[t_0,t_{\max})$. 

Proceeding similarly as in the proof of Theorem \ref{th:local}, we assume that $\lim_{t\to t_{\max}} h_v(x(t)) = 0$ and consider a constant $t'_v \in [t_0,t_{\max})$ such that $x(t) \in \mathcal{C}_{v,\mu'_v}$ for all $t\in [t'_v,t_{\max})$, implying $\sigma_{\mu_v}(h_v(x(t))) >  \sigma_{\mu_v}(\mu'_v) > 0$, for all $t\in [t'_v,t_{\max})$. By using \eqref{eq:control law square}, $\dot{\beta}_v$ becomes 
\begin{align*}
    \dot{\beta}_v = \beta_{v,\textup{d}} f_\textup{n}(x) - \kappa_v \sigma_{\mu_v}(h_v) \beta_{v,\textup{d}}^2 \nabla_{x_2} h_v(x)^\top g(x) \nabla_{x_2} h_v(x), 
\end{align*}
for all $t\in[t_v',t_{\max})$, where $f_\textup{n}(x) \coloneqq \nabla_{x_1} h_v(x)^\top x_2 + \nabla_{x_2} h_v(x)( f(x) + g(x)u_\textup{n}(x))$ is a continuous function bounded by a constant $\bar{f}_\textup{n}$, for all $t\in[t_v',t_{\max})$. 
By using the identity $g(x) = \frac{1}{2}(g(x) + g(x)^\top) + \frac{1}{2}(g(x) - g(x)^\top)$, and employing the skew symmetry of $g(x) - g(x)^\top$ and the positive definiteness of $g(x) + g(x)^\top$, we obtain 
\begin{align*}
    \dot{\beta}_v \leq |\beta_{v,\textup{d}}| \bar{f}_\textup{n} - \kappa_v \underline{g} \sigma_{\mu_v}(\mu'_v)\beta^2_{v,\textup{d}} \| \nabla_{x_2} h_v(x) \|^2,
\end{align*}
where $\underline{g}$ is the minimum eigenvalue of $g(x) + g(x)^\top$, which is positive for all $t\in[t_v',t_{\max})$. In addition, since $x(t) \in  \mathcal{C}_{\mu'_v} \subset \mathcal{C}_{\mu_v}$ for $t\in[t_v',t_{\max})$, it holds that $\|\nabla_{x_2} h_v(x)\| \geq \varepsilon_v >0$, which leads to 
\begin{align*}
    \dot{\beta}_v \leq -\epsilon_v \kappa_v \underline{g} \sigma_{\mu_v}(\mu'_v)|\beta_{v,\textup{d}}|\left(|\beta_{v,\textup{d}}| - \frac{\bar{f}_\textup{n} }{\epsilon_v\kappa_v \underline{g} \sigma_{\mu_v}(\mu'_v)} \right),
\end{align*}
for all $t\in[t_v',t_{\max})$. Therefore, we conclude that $\dot{\beta}_v > 0$ when $|\beta_{v,\textup{d}}| > \frac{\bar{f}_\textup{n} }{\varepsilon_v\kappa_v \underline{g} \sigma_{\mu_v}(\mu'_v)}$.  By following similar arguments as in the proof of Theorem \ref{th:local}, the proof follows.
\end{proof}

\begin{remark}[High relative degree] \label{rem:high-order dyn}
    Theorems \ref{th:local} and \ref{th:square} suggest a way to tackle higher $k$-order dynamics of the form 
    \begin{align*} 
	\dot{x}_i &= f_i(\bar{x}_i) + g_i(\bar{x}_i)x_{i+1}, \ i\in\{1,\dots,k-1\}, \\
	\dot{x}_k &= f(x_k) + g(x)u
    \end{align*}
    for a positive integer $k$, where $\bar{x}_i \coloneqq [x^\top_1,\dots,x_i^\top]^\top \in \mathbb{R}^{n \cdot i}$, for all $i\in\{1,\dots,k\}$, and $x\coloneqq \bar{x}_k$. By assuming that $g_i + g_i^\top$ are positive definite, for all $i\in\{1,\dots,k-1\}$, we can design continuous reference signals $x_{i+1,\textup{r}}$ for the states $x_{i+1}$, as in \eqref{eq:x_2des} and based on consecutive error signals as in \eqref{eq:e_2}. The control signal $u$ can then be designed based on the over-approximation of $g(x)$ and a reciprocal barrier on the difference $x_{k} - x_{k,\textup{r}}$, as in \eqref{eq:control law local}. 
\end{remark}

\section{Controllability loss} \label{sec:contr loss}

In this section we provide an algorithm that considers cases where $\|\hat{g}^i(x)^\top \nabla_{x_2} h_v(x)\|$ can become arbitrarily small, relaxing thus the assumption \eqref{eq:local cond1} in Theorem  \ref{th:local}. For technical requirements, we assume in the following that $\mathcal{C}_v$ is a compact $2n$-dimensional manifold. 

The framework we follow in order to tackle such cases is an online switching mechanism that computes locally alternative barrier functions $h_j$, defining new safe sets $\mathcal{C}_j \coloneqq \{ x\in\mathbb{R}^n: h_j(x) \geq 0\} \subset \mathcal{C}_{j-1}$, for $j \geq 2$, and $h_1 \coloneqq h_v$, $\mathcal{C}_1 \coloneqq \mathcal{C}_v$. 
More specifically, at a point $x_c$ where $\|\hat{g}^i(x_c)^\top \nabla_{x_2} h_v(x_c)\|$ becomes too small, the algorithm looks for an alternative function $h_2$ satisfying $\mathcal{C}_2 \subset \mathcal{C}_v$ and for which $\|\hat{g}^i(x_c)^\top \nabla_{x_2} h_2(x_c)\|$ is sufficiently large. 
For illustration purposes, consider the simple example of a system with $x_1 = [x_{11},x_{12}]^\top  \in \mathbb{R}^2$, $x_2 = [x_{21},x_{22}^\top]^\top \in \mathbb{R}^2$ with $\hat{g}^i(x) = [2x_{21}, -2]^\top$ for some $i$, and $h_v(x) = 4.5 - \|x_2\|^2$ representing a sphere with radius of $\sqrt{4.5}$ (depicted with red in Fig. \ref{fig:sets}). The term $\|\hat{g}^i(x)^\top \nabla_{x_2} h_v(x)\|$ vanishes on the parabola $x_{22} = x_{21}^2$ (depicted with black in Fig. \ref{fig:sets}). When $x$ is close enough to that line, the proposed algorithm computes a new $h_2$; Fig. \ref{fig:sets} depicts the case when $x_c = [1,0.95]$ (green asterisk). A potential choice is then the ellipsoidal set $h_2 = -0.175x_{21}^2 + 0.2x_{21}x_{22} - 0.2x_{22}^2 + 0.15x_{21} - 0.0889x_{22} + 0.2$ (depicted with blue in Fig. \ref{fig:sets}). The new $h_2$ satisfies $C_2 \subset \mathcal{C}_v$, $h_2(x_c) > 0$, while $\|\hat{g}^i(x)^\top \nabla_{x_2} h_2(x)\|$ vanishes on the parabola $0.1x_{21} - 0.8x_{22} - 0.4x_{21}x_{22} + 0.7x_{21}^2 - 0.1779$ (depicted with purple in Fig. \ref{fig:sets}), with $\|\hat{g}^i(x_c)^\top \nabla_{x_2} h_2(x_c)\| = 1.6379$. 
The controller switches locally to $h_2$ until $\|\hat{g}^i(x)^\top \nabla_{x_2} h_v(x)\|$ becomes sufficiently large again. In case the system navigates along the line $\|\hat{g}^i(x)^\top \nabla_{x_2} h_v(x)\|=0$ to a region where $\|\hat{g}^i(x)^\top \nabla_{x_2} h_2(x)\|$ is also small, the procedure is repeated and a new $h_3$ is computed.
In the example given, the region around the point where both $\|\hat{g}^i(x)^\top \nabla_{x_2} h_v(x)\|=0$ and $\|\hat{g}^i(x)^\top \nabla_{x_2} h_2(x)\| =0$ is around the intersection of the black and purple lines in Fig. \ref{fig:sets}. Note, however, that in the specific example the system cannot navigate close to that point, since employment of $h_2$ will keep it bounded in $\mathcal{C}_2$ (the set defined by the blue line in the figure). Similarly to $h_v(x)$, the functions $h_j$, $j\geq 2$, depend explicitly on $e_2$, i.e., the difference $x_2 - x_{2,\textup{r}}(x_1)$. 
Note that, since we desire $\hat{g}^i(x)$ to be close to $g(x)$, and hence the interval $\mathsf{G}^i(x)$ to be small (see Lemma \ref{lem:estimate}), choosing a different $\hat{g}^i(x)$ from $\mathsf{G}^i(x)$ is not expected to significantly modify the term $\|\hat{g}^i(x)^\top \nabla_{x_2} h_v (x)\|$.

The aforementioned procedure is described more formally in the algorithm $\mathsf{SafetyAdaptation}$ (Alg. \ref{alg:main alg}). More specifically, for a given $j$, each $\rho_j$ indicates whether the system is close to the set where $\|\hat{g}^i(x)^\top \nabla_{x_2} h_j(x) \| = 0$. If that's the case ($\rho_j = 0$), then a new function $h_{j+1}$ is computed such that $C_{j+1} \subset C_j$, and in the switching point it holds that $h_{j+1}(x) > 0$ and $\|\hat{g}^i(x)^\top \nabla_{x_2} h_{j+1}(x) \|$ is sufficiently large; $h_{j+1}(x)$ is used then in the control law.  If $\|\hat{g}^i(x)^\top \nabla_{x_2} h_\iota(x) \|$ becomes sufficiently large, for some $\iota < j$, then $j$ is set back to $\iota$, and $h_\iota(x)$ can be safely used in the control law again. We also impose a hysteresis mechanism for the switching of the constants $\rho_\iota$ (lines 4, 10) through the parameters $\bar{\varepsilon}$, $\underline{\varepsilon}$ (see Fig. \ref{fig:rho_hyst}). 

The reasoning behind Alg. \ref{alg:main alg} is the following. By appropriately choosing the functions $h_\iota(x)$, the solutions of $\|\hat{g}^i(x)^\top \nabla_{x_2} h_\iota(x)\| = 0$ form curves of measure zero. Hence, the intersection of $2n$ such lines will be a single point in $\mathbb{R}^{2n}$ and hence the employment of a newly computed $h_{2n+1}$ will drive the system away from that point, resetting the algorithm. More formally, there is no $t \geq t_0$ and $j \geq 2n + 1$ such that $\|g^i(x(t))^\top \nabla_{x_2} h_j(x(t))\| \leq \underline{\varepsilon}$, implying that the 
iterator variable $j$ of Alg. \ref{alg:main alg} will never exceed $2n+1$. 

Following similar steps as with $h$ and $h_v$, we define continuously differentiable reciprocal functions $\beta_j:(0,\infty)\to \mathbb{R}$ that satisfy
\begin{align} \label{eq:beta cond general}
	\frac{1}{\alpha_{j_1}(h_j)} \leq \beta_j(h_j) \leq \frac{1}{\alpha_{j_2}(h_j)} 
\end{align}
for class $\mathcal{K}$ functions $\alpha_{j_1}$, $\alpha_{j_2}$, and $j \in \{1,\dots,2n+1\}$. 
The formal definition of the control law is then
\begin{subequations}\label{eq:control law general}
\begin{align} 
	u &= u_\textup{n}(x) -  \kappa_v \sigma_{\mu_v}(h_v(x))  u_\textup{b}(x,t) \\ 
	u_\textup{b} &\coloneqq \sum_{\iota=1}^{2n+1}\rho_\iota \prod_{j=1}^{\iota-1} (1-\rho_j) \beta_{j,\textup{d}} \frac{\hat{g}^i(x)^\top \nabla_{x_2} h_\iota(x) }{\|\hat{g}^i(x)^\top \nabla_{x_2} h_\iota(x)\|^2},
\end{align}
\end{subequations}
for all $t\in[t_i,t_{i+1})$, $i\in\bar{\mathbb{N}}$,
with $h_1 = h_v$, $\beta_1 = \beta_v$, $\beta_{j,\textup{d}} \coloneqq \frac{\textup{d} \beta_j(h_j(x))}{\textup{d} h_j(x)}$, and $\sigma_{\mu_v}$ as in \eqref{eq:sigma mu}. 

\begin{figure}
	\centering
	\includegraphics[width=0.45\textwidth]{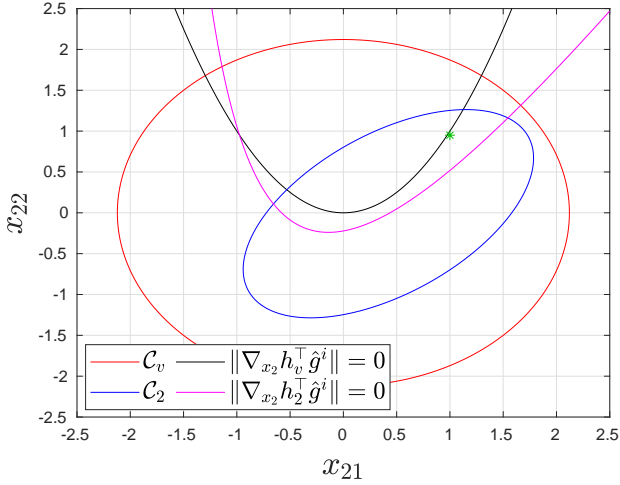}	
	\caption{Illustration of the adaptation algorithm (Algorithm \ref{alg:main alg}). At a point $x_c$ on the curve $\|g^i(x)^\top\nabla_{x_2}h_v(x)\| = 0$ (green point on black curve), the algorithm computes a new $h_2(x)$ such that $\mathcal{C}_2 \subset \mathcal{C}_v$ (blue curve) and $x_c$ is sufficiently far from the curve $\|g^i(x)^\top\nabla_{x_2}h_2(x)\| = 0$ (purple curve).   }
	\label{fig:sets}
\end{figure}

The $\mathsf{SafetyAdaptation}$ algorithm is run separately for each time interval $[t_i,t_{i+1})$, $i\in\bar{\mathbb{N}}$. That is, when a new measurement $(x^{i+1},\dot{x}^{i+1},u^{u+1})$ is received, the estimation of $g(x)$ is updated, a new $\hat{g}^{i+1}$ is computed by Alg. \ref{algo:overapprox-datapoints}, and $\mathsf{SafetyAdaptation}$ is reset ($j$ and $\rho_\iota$ are reset as in lines $1-3$). 
This is illustrated in  the algorithm $\mathsf{SafetyControl}$ (Alg. \ref{alg:run  alg}).


\begin{algorithm}
	\caption{$\mathsf{SafetyAdaptation}(g^i,h,\bar{\varepsilon},\underline{\varepsilon})$}\label{alg:main alg}
	\begin{algorithmic}[1]
		\State $\rho_\iota \leftarrow 1$, $\forall \iota \in \{1,\dots,n+1\}$;			
		\State $j \leftarrow 1$;  $h_1 \leftarrow h_v$;
		\While {$\mathsf{True}$}
				\If {$\|\hat{g}^i(x)^\top \nabla_{x_2} h_j(x)\| \leq \underline{\varepsilon} \land \rho_j = 1$ }
				\State $x_c \leftarrow x$;     $\rho_j \leftarrow 0$;
				\State Find $h_{j+1}$ such that 
				\begin{enumerate}
					\item $\mathcal{C}(h_{j+1}) \subset \mathcal{C}(h_j)$ 
					\item $h_{j+1}(x_c) > 0$
					\item $\|\hat{g}^i(x_c)^\top \nabla_{x_2} h_{j+1}(x_c)\| \geq \gamma_{j+1}$;
				\end{enumerate}
				\State $j \leftarrow j+1$;
				\EndIf
			\For {$\iota\in\{1,\dots,j-1\}$}
				\If {$\|\hat{g}^i(x)^\top \nabla_{x_2} h_\iota(x)\| > \bar{\varepsilon} \land \rho_\iota = 0$}
						\State $\rho_\iota \leftarrow 1$; $j \leftarrow i$;				
						\State Break;
				\EndIf
			\EndFor			
			\State Apply \eqref{eq:control law general}
		\EndWhile
		
	\end{algorithmic}
\end{algorithm}


Algorithm \ref{alg:main alg} imposes an extra, state-dependent switching to the closed-loop system, which can be written as
\begin{align} \label{eq:closedloop system}
\dot{x} = f(x) + g(x)u_\textup{n} - \kappa_v\sigma_{\mu_v}(h_v(x))   g(x)   u_j(x,t),
\end{align}
where $u_j \coloneqq \beta_{j,\textup{d}}(x)\frac{\hat{g}^i(x)^\top \nabla_{x_2}
h_j(x)}{\|\hat{g}^i(x)^\top \nabla_{x_2} h_j(x)\|^2}$, for some $i \in \bar{\mathbb{N}}$, $j \in \mathbb{N}$. 
The switching regions are not pre-defined, but detected online based on the trajectory of the system (line 4 of Alg. \ref{alg:main alg}). Moreover, by choosing $\gamma_j > \underline{\varepsilon}$, for all $j\geq 2$, we guarantee that the switching does not happen continuously, and hence the solution of \eqref{eq:closedloop system} is well-defined in $[t_0,t_{\max})$ for some $t_{\max} > t_0$, satisfying $x(t) \in \textup{Int}(\bar{\mathcal{C}})\cap  \textup{Int}(\mathcal{C}_v)$, for all $t\in[t_0,t_{\max})$ (similar to the proof of Lemma \ref{lem:local solution}).

\begin{figure}
	\centering
	\includegraphics[width=0.4\textwidth]{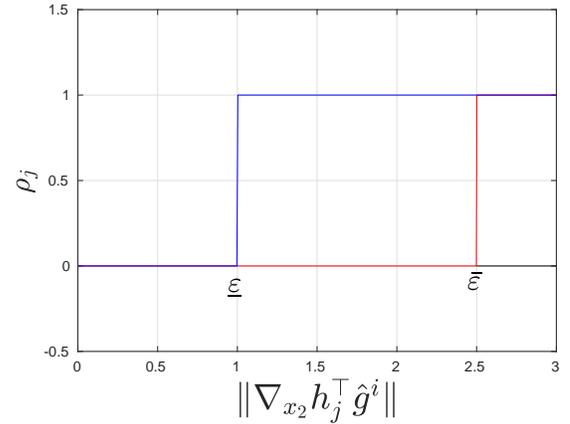}	
	\caption{Depiction of the hysteresis mechanism for avoidance of continuous switching; $\rho_j(t)$ is changed to $0$ if $\rho_j(t^-) = 1$ and $\|\hat{g}^i(x(t)) \nabla_{x_2}h_j(x(t)) \| \leq \underline{\varepsilon}$ and to $1$ if 
	$\rho_j(t^-) = 0$ and $\|\hat{g}^j(x(t)) \nabla_{x_2}h_j(x(t)) \| \geq \bar{\varepsilon}$.}
	\label{fig:rho_hyst}
\end{figure}

\begin{algorithm}
	\caption{$\mathsf{SafetyControl}(\bar{\varepsilon},\underline{\varepsilon})$}
	\label{alg:run  alg}
	\begin{algorithmic}[1]
		\State $i \leftarrow 1$;
		\State $\hat{g}^i \leftarrow \mathsf{rand}(n,m)$;
		\While{$\mathsf{True}$}
			\While{$\mathsf{NoNewMeasurement}$}
				\State $\mathsf{SafetyAdaptation(}\hat{g}^i,h,\bar{\varepsilon},\underline{\varepsilon})$			
			\EndWhile
			\State $i \leftarrow i+ 1$;
			\State $\hat{g}^i \leftarrow \mathsf{Approximate}(\dot{x}(t_i),x(t_i),u(t_i))$;
		\EndWhile
	\end{algorithmic}
\end{algorithm}

In the following, we prove that, for each $i\in\bar{\mathbb{N}}$, the iterator $j$ in Algorithm \ref{alg:main alg}, indicating the number of $h_j$ computed, reaches at most $2n+1$. 

For $i\in\bar{\mathbb{N}}$, let the sets $\mathcal{S}^i_j \coloneqq \{x \in \bar{\mathcal{C}} \cap \textup{Cl}(\mathcal{C}_{v,\mu_v}): \hat{g}^i(x)^\top \nabla_{x_2} h_j(x) = 0\}$, as well as the sets ${\mathfrak{S}}^i_j  \coloneqq\bigcap^j_{q=1} \mathcal{S}^i_q$, for all $j\in\{1,\dots,2n\}$. 
The following assumption is needed for the subsequent analysis:
\begin{assumption} \label{ass:dim}
	Let $i\in\bar{\mathbb{N}}$. The sets $\mathcal{S}_j^i$, ${\mathfrak{S}}^i_j$ are manifolds satisfying $\textup{dim}(\mathcal{S}_j^i) \leq 2n - 1$, $\forall j\in\{1,\dots,2n\}$, and $\textup{codim}( {\mathfrak{S}}^i_j \cap \mathcal{S}^i_{j+1}) \geq \textup{codim}({\mathfrak{S}}^i_j) + \textup{codim}(\mathcal{S}^i_{j+1})$, $\forall j\in\{1,\dots,2n-1\}$.
\end{assumption}

Assumption \ref{ass:dim} first  states that $\mathcal{S}_j^i$ are lower-dimension manifolds (e.g., lines on the plane or planes in $3$D space), which is common for curves like $\hat{g}^i(x)^\top \nabla_{x_2} h_j(x) = 0$. Additionally, the dimension condition 
is essentially a mild transversality condition on the manifolds $\mathcal{S}^i_j$ \cite{guillemin2010differential}. Note that Assumption \ref{ass:dim} implies that $\textup{dim}( {\mathfrak{S}}^i_j \cap \mathcal{S}^i_{j+1}) \leq \textup{dim}(\mathfrak{S}^i_j) + \textup{dim}(\mathcal{S}^i_{j+1}) - 2n$.

We also define the inflated sets $\widetilde{S}^i_j(\bar{\varepsilon}) \coloneqq \{ x \in \bar{\mathcal{C}} \cap \textup{Cl}(\mathcal{C}_{v,\mu_v}) :\| \hat{g}^i(x)^\top \nabla_{x_2} h_j(x) \| \leq \bar{\varepsilon}\}$ and the respective intersections $\widetilde{\mathfrak{S}}^i_j(\bar{\varepsilon}) \coloneqq \bigcap_{q=1}^j \widetilde{\mathcal{S}}^i_q(\bar{\varepsilon})$,
 $j\in\{1,\dots,2n\}$.
 
 
 We note that, based on the defined maximal solution, $x(t)$ evolves in 
 $\textup{Int}(\bar{\mathcal{C}})\cap \mathcal{C}_{v,\mu_v} \subset \bar{\mathcal{C}}\cap  \textup{Cl}(\mathcal{C}_{v,\mu_v})$ for $t\in[t_0,t_{\max})$. Nevertheless, we employ the closed set $\bar{\mathcal{C}}\cap \textup{Cl}(\mathcal{C}_{v,\mu_v})$ in the aforementioned definitions to ease the following technical presentation and avoid unnecessary notational jargon.

The solution $x(t)$, defined for $t\in [t_0,t_{\max})$, can be then decomposed based on the update times $t_i$, $i\in\bar{\mathbb{N}}$ as 
\begin{align*}
	x(t) = 
		\begin{cases}
		x^0(t), \ \ t\in [t_0,t_1)\\
		\vdots \\
		x^p(t), \ \ t\in [t_p, t_{\max}), \\
		\end{cases}
\end{align*}
for some $p\in\bar{\mathbb{N}}$. 
Note that, similar to \eqref{eq:control law local}, \eqref{eq:control law general} implies that the safety controller is activated close to the boundary of $\mathcal{C}_v$, i.e., in $\mathcal{C}_{v,\mu_v}$. Similarly to Section \ref{sec:control design}, we use the set 
$K_{\mu_v} = \{ i \in \bar{\mathbb{N}} : \exists [\tau_1,\tau_2) \subseteq [t_0,t_{\max}), \text{ with } \tau_1 \in [t_i, t_{i+1})\text{ s.t. } x(t) \in \textup{Int}(\bar{\mathcal{C}}) \cap \mathcal{C}_{v,\mu_v}, \forall t\in[\tau_1,\tau_2) \}$. 
In what follows, we focus on the solution parts $x^i(t)$, for $i\in K_{\mu_v}$. In particular, let a fixed $i \in K_{\mu_v}$, corresponding to the time interval $[t_{i},t_{i+1})$. Let also $\tau_{i} \coloneqq \inf\{ t \in [t_{i},t_{i+1}) : x(t) \in \textup{Int}(\bar{\mathcal{C}})\cap \mathcal{C}_{v,\mu_v}\}$, which is well-defined in view of the definition of $K_{\mu_v}$. 


Let the solution restriction $\widetilde{x}_i(t) \coloneqq x(t)$ for $t\in [\tau_i, t_{i+1})$.  Then it holds that $\widetilde{x}_i(t) \in \textup{Int}(\bar{\mathcal{C}})\cap \mathcal{C}_{v,\mu_v}$, and $\widetilde{x}_i(t)$ can be further decomposed 
based on the iterator $j$ of Algorithm \ref{alg:main alg} as 
\begin{align} \label{eq:sol decomp}
\widetilde{x}_i(t) =
	\begin{cases}		
	& y_{1,I_1}(t), \ \ t \in T_{1,1} = [\mathfrak{l}_{1,I_1},\mathfrak{r}_{1,I_1}),  \\
		  &	y_{1,I_1+1}(t), \ \ t \in T_{1,I_1+1} = [\mathfrak{l}_{1,I_1+1},\mathfrak{r}_{1,I_1+1}), \\
		  & \dots \\
		  & y_{1,{F_1}}(t), \ \ t \in T_{1,F_1} = [\mathfrak{l}_{1,F_1},\mathfrak{r}_{1,F_1}), \\		  
		  & y_{2,{I_2}}(t), \ \ t \in T_{2,I_2} = [\mathfrak{l}_{2,I_2},\mathfrak{r}_{2,I_2}), \\		  
		  & y_{2,{I_2+1}}(t), \ \ t \in T_{2,I_2+1} = [\mathfrak{l}_{2,I_2+1},\mathfrak{r}_{2,I_2+1}), \\ 
		  &\dots \\		  
		  & y_{2,{F_2}}(t), \ \ t \in T_{2,F_2} = [\mathfrak{l}_{2,F_2},\mathfrak{r}_{2,F_2}), \\		  
		  & y_{3,{I_3}}(t), \ \ t \in T_{3,I_3} = [\mathfrak{l}_{3,I_3},\mathfrak{r}_{3,I_3}), \\		  
  		  & y_{3,{I_3+1}}(t), \ \ t \in T_{3,I_3+1} = [\mathfrak{l}_{3,I_3+1},\mathfrak{r}_{3,I_3+1}), \\		  
  		  &  \dots 
	  \end{cases}
\end{align}
with $\mathfrak{l}_{1,I_1} = \tau_i$, $I_1=1$. 
The signal $y_{\ell,j}(t)$ stands for the solution when $h_j(x)$ is active for the $\ell$th time (note the reset of the $\mathsf{SafetyAdaptation}$ algorithm in line 15). 
Note also that $\ell$ is finite due to the hysteresis mechanism and the fact that $\gamma_j > \underline{\varepsilon}$; We denote by $\bar{\ell}$ its maximum value. 
The indices $F_\ell$, $I_\ell$ are the last and first values of $j$ (defining $h_j(x)$) at the $\ell$th time (with $I_1=1$).  
The respective time intervals are defined as
\begin{align*}
	 \mathfrak{r}_{\ell,j} &= \mathfrak{l}_{\ell,j+1} \\
	&\coloneqq  \inf\{ t \in T_{\ell,j}  : \| \hat{g}^i(x(t))^\top \nabla_{x_2} h_{j}(x(t)) \| \leq \underline{\varepsilon} \}, \\ 
	&\hspace{20mm} \forall j \in\{I_\ell,\dots,F_\ell-1\} \\
	\mathfrak{r}_{\ell,F_\ell} &= \mathfrak{l}_{\ell+1,I_{\ell+1}} \coloneqq \inf\{ t \in T_{\ell,F_{\ell}} : \\ 
	& \exists \iota \in \{1,\dots,F_{\ell}\} \text{ s.t. } \| \hat{g}^i(x(t))^\top \nabla_{x_2} h_\iota(x(t)) \| > \bar{\varepsilon}  \}, 
\end{align*}
with $\mathfrak{l}_{1,I_1} = \tau_i$, $I_1 = 1$. Note that $\iota$ from the definition of $\mathfrak{r}_{\ell,F_\ell}$ is equal to $I_{\ell + 1}$ (see line 11 of Algorithm \ref{alg:main alg}).

It holds that $T_{\ell,j} \subset [\tau_i,t_{i+1})$ (or $T_{\ell,j} \subset [\tau_i,t_{\max})$ in case $i=p$), for all $\ell \in\{1,\dots,\bar{\ell}\}$, $j\in[I_\ell,\dots,F_\ell]$.
Moreover, it holds  
that $y_{\ell,j}(t) \in \widetilde{\mathfrak{S}}^i_{j-1}(\bar{\varepsilon}) $, $\forall t\in T_{\ell,j}$, for all $\ell \in\{1,\dots,\bar{\ell}\}$, $j \in\{I_\ell, \dots,F_\ell\}$, where
$\widetilde{\mathfrak{S}}^i_{0}(\bar{\varepsilon}) \coloneqq   \bar{\mathcal{C}}\cap \textup{Cl}(\mathcal{C}_{v,\mu_v}) \backslash \bigcup_{j\in\{1,\dots,2n\}}\widetilde{\mathcal{S}}^i_{j}(\bar{\varepsilon})$.

In view of Assumption \ref{ass:dim} and since $\textup{Int}(\bar{\mathcal{C}})\cap \mathcal{C}_{v,\mu_v}$ is bounded, the sets $\widetilde{\mathfrak{S}}^i_j(\bar{\varepsilon})$ 
are constituted by the union of a finite number (at least $1$) of connected components, where $\|\hat{g}^i(x)^\top \nabla_{x_2} h_\iota(x)\| \leq \bar{\varepsilon}$ holds for all $\iota \in\{1,\dots,j\}$ i.e., 
$\widetilde{\mathfrak{S}}^i_j(\bar{\varepsilon}) \coloneqq \bigcup_{q \in \mathcal{L}^i_j} \widetilde{\mathcal{K}}^{i,j}_q(\bar{\varepsilon})$ for some finite index set $\mathcal{L}^i_j \subset \mathbb{N}$, for $j\in\{1,\dots,2n\}$. Since $\widetilde{\mathcal{K}}^{i,j}_{q}$ are closed, there exist a $\lambda^\ast$ such that $\widetilde{\mathcal{K}}^{i,j}_{q_1}(\lambda^\ast)\cap \widetilde{\mathcal{K}}^{i,j}_{q_2}(\lambda^\ast) = \emptyset$, for all $q_1,q_2 \in \mathcal{L}^i_j$ with $q_1\neq q_2$, $j\in\{1,\dots,2n\}$.


We show now that, by choosing a small enough $\bar{\varepsilon}$, each trajectory part $y_{\ell,j}(t)$ lies only in one of the $\widetilde{\mathcal{K}}^{\ell,j}_q$, 
for $j\in\{1,\dots,2n\}$.


\begin{proposition} \label{prop:epsilon-lambda}
	Let $\ell \in\{1,\dots,\bar{\ell}\}$, $j \in \{I_\ell,\dots,F_\ell\}$, $I_\ell \geq 1$ and assume that $F_\ell \leq 2n+1$. 
	Then the choice $\bar{\varepsilon} < \frac{\lambda^\ast}{\sqrt{2n}}$ guarantees that there exists a $q^\ast \in \mathcal{L}^i_{j-1}$ such that $y_{\ell,j}(t) \in \widetilde{\mathcal{K}}^{i,j-1}_{q^\ast}(\lambda^\ast)$, implying $y_{\ell,j}(t) \notin \widetilde{\mathcal{K}}^{i,j-1}_{q}(\lambda^\ast)$, $\forall q \in \mathcal{L}^i_{j-1}\backslash \{q^\ast\}$.
\end{proposition}
\begin{proof}	
	Note first that it holds that $\sum_{\iota=1}^{j-1} \| \hat{g}^i(y_{\ell,j}(t))^\top \nabla_{x_2} h_\iota(y_{\ell,j}(t))  \|^2 \leq (j-1)\bar{\varepsilon}^2$, 
	since the latter forms the circumscribed circle of the cube 
	$\widetilde{\mathfrak{S}}^i_{j-1}(\bar{\varepsilon}) = \bigcap_{q=1}^{j-1} \widetilde{\mathcal{S}}^i_q(\bar{\varepsilon})$.
	 Since $j \leq 2n+1$, by choosing $\bar{\varepsilon} \leq \frac{\lambda^\ast}{\sqrt{2n}} \leq \frac{\lambda}{\sqrt{j-1}}$, we guarantee that $\sum_{\iota=1}^{j-1} \| \hat{g}(y_{\ell,j}(t))^\top \nabla_{x_2} h_\iota(y_{\ell,j}(t)) \|^2 \leq (\lambda^\ast)^2$, which is the inscribed circle of the cube $\widetilde{\mathfrak{S}}^i_{j-1}(\lambda^\ast)$, 
	 which implies that $\| \hat{g}^i(y_{\ell,j}(t))^\top \nabla_{x_2} h_\iota(y_{\ell,j}(t)) \| \leq \lambda^\ast$, for all $ \iota\in\{1,\dots,j-1\}$. Since the sets $\widetilde{\mathcal{K}}^{i,j-1}_q(\lambda^\ast)$ are disjoint, we conclude that $y_{\ell,j}(t)$ belongs to only one $\widetilde{\mathcal{K}}^{i,j-1}_{q^\ast}(\lambda^\ast)$, for some $q^\ast \in \mathcal{L}^i_{j-1}$, and $y_{\ell,j}(t) \notin \widetilde{\mathcal{K}}^{i,j-1}_{q}(\lambda^\ast)$, $\forall q \in \mathcal{L}^i_{j-1}\backslash \{q^\ast\}$.
	
	
	 
	
\end{proof}

By using Proposition \ref{prop:epsilon-lambda} and Assumption \ref{ass:dim},
we prove next that by choosing a sufficiently large $\gamma_{n+1}$, we guarantee that the iterator variable $j$ of Algorithm \ref{alg:main alg} is bounded by $2n+1$. 

\begin{proposition}
	There exist positive constants $\gamma$, $\omega$ such that if $\bar{\varepsilon} < \omega$ and  $\gamma_{n+1} \geq \gamma$, 
	there are no $t \geq t_0$ and $j \geq 2n + 1$ such that $\|g^i(x(t))^\top \nabla_{x_2} h_j(x(t))\| \leq \underline{\varepsilon}$.
\end{proposition}
\begin{proof}
Let $j=2n+1$ in Algorithm \ref{alg:main alg} for some $\ell \in\{1,\dots,\bar{\ell}\}$, i.e., 
\begin{equation} \label{eq:x(t) k=n+1}
	x(t) = y_{\ell,2n+1}(t) \in \widetilde{\mathfrak{S}}^i_{2n}(\bar{\varepsilon})  =  \bigcap_{q=1}^{2n} \widetilde{\mathcal{S}}^i_q(\bar{\varepsilon}), \ \ t \in T_{\ell,2n} 
\end{equation}
Assume that ${\mathfrak{S}}^i_{2n}  = \emptyset$. 
Since $\mathcal{S}^i_q$ are closed, it can be concluded that there exists a positive constant $\omega$ such that $\widetilde{\mathfrak{S}}^i_{2n}({\omega})  = \emptyset$. Hence, by choosing $\bar{\varepsilon} \leq \omega$, we guarantee that $\widetilde{\mathfrak{S}}^i_{2n}(\bar{\varepsilon})=\emptyset$, which contradicts \eqref{eq:x(t) k=n+1}. Hence, we conclude that $\mathfrak{S}^i_{2n}\neq \emptyset$, which, in view of Assumption \ref{ass:dim}, implies that $\text{dim}(\mathfrak{S}^i_{2n}) = 0$. Therefore, the set $\mathfrak{S}^i_{2n}$ is a zero-dimensional manifold consisting of a finite set of points $\{p_1,p_2,\dots,p_m\}$ of $\mathbb{R}^{2n}$ for some $m\in\mathbb{N}$. The set $\widetilde{\mathfrak{S}}^i_{2n}(\bar{\varepsilon})$ is the intersection of $\bar{\mathcal{C}}\cap\textup{Cl}(\mathcal{C}_{v,\mu_v})$
with a union of closed hypercubes around these points. In particular, based on the discussion prior  to Prop. \ref{prop:epsilon-lambda}, $\widetilde{\mathfrak{S}}^i_{2n}(\bar{\varepsilon}) =  \bigcup_{q\in \mathcal{L}^i_{2n}} \widetilde{K}^{i,{2n}}_q(\bar{\varepsilon})$, where $\widetilde{\mathcal{K}}^{i,{2n}}_q(\bar{\varepsilon})$ are the intersections of these closed hypercubes with $\bar{\mathcal{C}}\cap\textup{Cl}(\mathcal{C}_{v,\mu_v})$. According to Prop. \ref{prop:epsilon-lambda}, by choosing $\bar{\varepsilon}$ small enough, $\widetilde{K}^{i,{2n}}_q(\bar{\varepsilon})$ are disjoint, and hence $y_{\ell,2n+1}(t)$ evolves in the intersection of $\bar{\mathcal{C}}\cap\textup{Cl}(\mathcal{C}_{v,\mu_v})$ with the hypercube around one $p_\eta$, for some $\eta \in \{1,\dots,m\}$.
By considering the circumscribed circle of the hypercube, we conclude that $y_{\ell,2n+1}(t)$ evolves in the intersection of $\bar{\mathcal{C}}\cap\textup{Cl}(\mathcal{C}_{v,\mu_v})$ with the closed ball defined by
$\sum_{\iota=1}^{n} \| \hat{g}^i(y_{\ell,2n+1}(t))^\top \nabla h_\iota(y_{\ell,2n+1}(t))  \|^2 \leq {2n}\bar{\varepsilon}^2$, which is the closed ball $\bar{\mathcal{B}}(p_\eta,\bar{\varepsilon}\sqrt{2n})$.

Let now $x_c \coloneqq y_{\ell,2n+1}(\mathfrak{l}_{\ell,2n+1}) = y_{\ell,2n}(\mathfrak{r}_{\ell,2n})$, i.e., the first point when $\|\hat{g}^i(y_{\ell,2n}(t)) \nabla_{x_2} h_{2n}(y_{\ell,2n}(t)) \| \leq \underline{\varepsilon}$ occurs, where it holds that $\|\hat{g}^i(x_c) \nabla_{x_2} h_{2n+1}(x_c) \| \geq \gamma_{2n+1}$. 
Consider now $x \in \bar{\mathcal{B}}(p_\eta,\bar{\varepsilon}\sqrt{2n})$, representing the solution $y_{\ell,2n+1}$. 
By adding and subtracting $\|\hat{g}^i(x) \nabla_{x_2} h_{2n+1}(x) \|$, we obtain 
\small
\begin{multline*}
	\|\hat{g}^i(x_c) \nabla_{x_2} h_{2n+1}(x_c) \pm \hat{g}^i(x) \nabla_{x_2} h_{2n+1}(x) \|
	\leq \\ \|\hat{g}^i(x_c) \nabla_{x_2} h_{2n+1}(x_c) - \hat{g}^i(x) \nabla_{x_2} h_{2n+1}(x) \|  + 
	\|\hat{g}^i(x) \nabla_{x_2} h_{2n+1}(x) \|  \\ \leq  
	\widetilde{L}_{i,2n+1}\|x_c - x\| +  \|\hat{g}^i(x) \nabla_{x_2} h_{2n+1}(x) \|,
\end{multline*}
\normalsize
where $\widetilde{L}_{i,2n+1}$ is the Lipschitz constant of the function $\hat{g}^i(x) \nabla_{x_2} h_{2n+1}(x)$ in $\bar{\mathcal{B}}(p_\eta,\bar{\varepsilon}\sqrt{2n})$. Since $x, x_c\in \bar{\mathcal{B}}(p_\eta,\bar{\varepsilon}\sqrt{2n})$, it holds that $\|x_c - x\| \leq 2\bar{\varepsilon}\sqrt{2n}$.  By also using $\|\hat{g}^i(x_c) \nabla_{x_2} h_{2n+1}(x_c) \| \geq \gamma_{2n+1}$, we obtain 
\begin{align*}
	\gamma_{2n+1} \leq 2\widetilde{L}_{i,2n+1}\bar{\varepsilon}\sqrt{2n} + \|\hat{g}^i(x) \nabla_{x_2} h_{2n+1}(x) \|. 
\end{align*}
By choosing $\gamma_{2n+1} \geq \gamma \coloneqq 2\widetilde{L}_{k,n+1}\bar{\varepsilon}\sqrt{2n} + \underline{\varepsilon} + \chi$, where $\chi$ is an arbitrary positive constant, we guarantee that $\|\hat{g}^i(x) \nabla_{x_2} h_{2n+1}(x) \| \geq \underline{\varepsilon} + \chi$, for all $x \in \bar{\mathcal{B}}(p_\eta,\bar{\varepsilon}\sqrt{2n})$. Therefore, since $\widetilde{\mathfrak{S}}^i_{2n}(\bar{\varepsilon}) \subset \bar{\mathcal{B}}(p_\eta,\bar{\varepsilon}\sqrt{2n})$,  it holds that  $\widetilde{\mathfrak{S}}^i_{2n}(\bar{\varepsilon}) \bigcap \widetilde{\mathcal{S}}^i_{2n+1}(\bar{\varepsilon}) = \emptyset$, implying that the condition of line $4$ in Algorithm \ref{alg:main alg}, which would lead to $j = 2n+2$, cannot be satisfied when $j = 2n+1$, leading to the conclusion of the proof.
\end{proof}

We are now ready to state the main result of this section. 

\begin{theorem} \label{th:final}
	Let a system evolve according to \eqref{eq:system}, \eqref{eq:control law general} and a set $\mathcal{C}$, satisfying $x_1(t_0)\in \textup{Int}(\mathcal{C})$ for a positive $t_0 \geq 0$.  Let a constant $\mu_v' \in (0,\mu_v)$ and assume there exists a positive constant $\varepsilon$ such that 
	\begin{align*}
	    \bar{g}_{\mu'_v} \bar{H}_{v,\mu'_v} < \varepsilon \sigma_{\mu_v}(\mu'_v),
	\end{align*}
	where $\bar{g}_{\mu_v'} \coloneqq \sup_{\substack{x\in \mathcal{C}_{v,\mu_v'}\\i\in K_{\mu_v'} }} \|\widetilde{g}^i(x)\|$, and $\bar{H}_{v,\mu'_v} \coloneqq \max_{\iota\in\{1,\dots,2n+1\}}\{ \sup_{x\in\mathcal{C}_{v,\mu'_v}}\{ \|\nabla_{x_2} h_\iota(x)\| \}\}$. Under Assumptions \ref{ass:C in ball}, \ref{ass:grad h}, and \ref{ass:dim}, it holds that $x_1(t) \in \textup{Int}(\mathcal{C})$, and all closed-loop signals are bounded, for all $t\geq t_0$.
\end{theorem}

\begin{proof}
By following the arguments of the first part of the proof of Theorem \ref{th:local}, we can obtain the boundedness of $\beta(h(x_1(t))) \leq \bar{\beta}$ for some constant $\bar{\beta}$ and $t\in[t_0,t_{\max})$, implying the boundedness of $x_1(t)$ in a compact set $\widetilde{\mathcal{C}} \subset \textup{Int}(\mathcal{C})$, and the boundedness of $x_{2,\textup{r}}(x_1(t))$, $e_2(x(t))$, and $\dot{x}_{2,r}(x_1(t))$ for all $t\in[t_0,t_{\max})$. 
Next, assume that $t_{\max}$ is finite and $\lim_{t\to t_{\max}}h_v(x(t)) = 0$, implying $\lim_{t\to t_{\max}}\beta_v(h_v(x(t))) = \infty$, which we aim to contradict. 

Let $\bar{i} \coloneqq \max\{i \in \bar{\mathbb{N}} : t_i < t_{\max} \}$, and let $t' \coloneqq \inf\{ t'' \geq t_{\bar{i}} : x(t) \in \mathcal{C}_{v,\mu'_v}, \forall t\in[t'',t_{\max}) \}$. 
Then it holds that $x(t) \in \mathcal{C}_{v,\mu'_v}$ for all $t \in [t',t_{\max})$, and $x(t) \in \widetilde{\mathcal{C}}_v \coloneqq \{ x\in\mathbb{R}^{2n}: h_v(x) \geq \min_{t\in[t_0,t']}\{h(x(t))\} > 0 \}$, for all $t\in[t_0,t']$. Moreover, note that $\bar{i}\in K_{v,\mu'_v}$ and $\sigma_{\mu_v}(h_v(x(t))) >  \sigma_{\mu_v}(\mu'_v) > 0$, for all $t\in [t',t_{\max})$. 

Let the solution $x(t) = x^{\bar{i}}(t)$ for $t\in[t',t_{\max})$, be decomposed, similarly to \eqref{eq:sol decomp}, as
\begin{align*}
x^{\bar{i}}(t) =
\begin{cases}		
& y_{1,I_1}(t), \ \ t \in T_{1,1} = [t',\mathfrak{r}_{1,{I_1}}),  \\
& \dots \\
& y_{1,{F_1}}(t), \ \ t \in T_{1,F_1} = [\mathfrak{l}_{1,F_1},\mathfrak{r}_{1,F_1}), \\		  
& y_{2,{I_2}}(t), \ \ t \in T_{2,I_2} = [\mathfrak{l}_{2,I_2},\mathfrak{r}_{2,I_2}), \\		  
&\dots \\		  
& y_{2,{F_2}}(t), \ \ t \in T_{2,F_2} = [\mathfrak{l}_{2,F_2},\mathfrak{r}_{2,F_2}), \\		  
&  \dots 
\end{cases}
\end{align*}
with $I_\ell \leq F_\ell$ $\in \{1,\dots,2n+1\}$, for all $\ell \in\{1,\dots,\bar{\ell}\}$. Moreover, $T_{\ell,j} \subseteq [t',t_{\max})$, $j\in\{I_\ell,\dots,F_\ell\}$. 
The functions $\beta_j$, defined in \eqref{eq:beta cond general}, satisfy 
\begin{align*}
	    \dot{\beta}_j =& \beta_{j,\textup{d}} f_{j,\textup{n}}(y_{\ell,j}) -\\ 
	    &\hspace{-5mm} \kappa_v \sigma_{\mu_v}(h_v) \beta_{j,\textup{d}}^2 \nabla_{x_2} h_j(y_{\ell,j})^\top g(y_{\ell,j}) \frac{\hat{g}^i(y_{\ell,j})^\top \nabla_{x_2} h_v(y_{\ell,j}))}{ \|\hat{g}^i(y_{\ell,j})^\top \nabla_{x_2} h_v(y_{\ell,j}))\|^2} \\
	    \leq &|\beta_{j,\textup{d}}| \bar{f}_{j,\textup{n}} \\
	    &-\kappa_v \sigma_{\mu_v}(\mu'_v) \beta_{j,\textup{d}}^2   + \kappa_v \beta_{j,\textup{d}}^2 \frac{\widetilde{g}^i(y_{\ell,j})^\top \nabla_{x_2} h_j(y_{\ell,j})}{ \|\hat{g}^i(y_{\ell,j})^\top \nabla_{x_2} h_j(y_{\ell,j})\|} 
	\end{align*}
for all $t\in T_{\ell,j}$, where $f_{j,\textup{n}}(x) \coloneqq \nabla_{x_1} h_j(\mathsf{x})^\top x_2 + \nabla_{x_2} h_j(x)( f(x) + g(x)u_\textup{n}(x))$, for which it holds $\|f_{j,\textup{n}}(y_{\ell,j}(t))\| \leq \bar{f}_{j,\textup{n}}$ due to the continuity of $h_j(\cdot)$, $f(\cdot)$, $g(\cdot)$, $u_\textup{n}(\cdot)$, and the boundedness of $y(\ell,j)$ for all $t\in T_{\ell,j}$, similarly to the proof of Theorem \ref{th:local}.
In addition, since $\sup_{x\in \mathcal{C}_{v,\mu'_v}} \|\nabla_{x_2} h_j(x)\| \leq \bar{H}_{v,\mu'_v}$, we obtain 
\begin{align*}
    \dot{\beta}_j \leq& |\beta_{j,\textup{d}}| \bar{f}_{j,\textup{n}} 
     -\kappa_v \sigma_{\mu_v}(\mu'_v) \beta_{j,\textup{d}}^2 + \kappa_v \beta_{j,\textup{d}}^2 \frac{\bar{g}_{\mu'_v} \bar{H}_{v,\mu'_v}}{\underline{\varepsilon}},
\end{align*}
for all $t\in T_{\ell,j}$. By setting $\epsilon_{v} \coloneqq \sigma_{\mu_v}(\mu'_v) - \frac{\bar{g}_{\mu'_v} \bar{H}_{v,\mu'_v}}{\underline{\varepsilon}}$, which is positive for $t \in T_{\ell,j}$, we obtain 
\begin{align*}
	\dot{\beta}_j &\leq  - \kappa_v \epsilon_{v} |\beta_{j,\textup{d}}| \left( |\beta_{j,\textup{d}}|  - \frac{\bar{f}_{j,\textup{n}} }{\kappa_v \epsilon_v} \right), 
\end{align*}
for all $t\in T_{\ell,j}$. Therefore, it holds that $\dot{\beta}_j < 0$ when $|\beta_{j,\textup{d}}| > \sqrt{ \frac{\bar{f}_{j,\textup{n}} }{\kappa_v \epsilon_v} }$. By invoking similar argument as in the proof of Theorem \ref{th:local}, we conclude the boundedness of $\beta_j(h_j(y_{\ell,j}(t))) < \bar{\beta}_j$ for all $t\in T_{\ell,j}$.
Note that the aforementioned result holds for any $j\in\{I_\ell,\dots,F_\ell\}$. 
At the switching time instants $\mathfrak{l}_{\ell,j} = \mathfrak{r}_{\ell,j-1}$ it holds that $h_j(x(\mathfrak{l}_{\ell,j})) > 0$ and hence the functions $\beta_j$ are well-defined. 
Since $h_j(\cdot) < h_v(\cdot)$, we conclude that $\beta_v(y_{\ell,j}(t)) < \bar{\beta}_j(h_j(y_{\ell,j}(t)))$, for all $t \in T_{\ell,j}$, $j\in\{I_\ell,\dots,F_\ell\}$, which implies that there exists a finite constant $\bar{\beta}_v$ such that $\beta_v(h_v(x(t))) \leq \bar{\beta}_v$, for all $t\in [t',t_{\max})$, which contradicts $\lim_{t\to t_{\max}} \beta_v(h_v(x(t))) = \infty$.  
\begin{figure}
	\centering
	\includegraphics[width=0.45\textwidth]{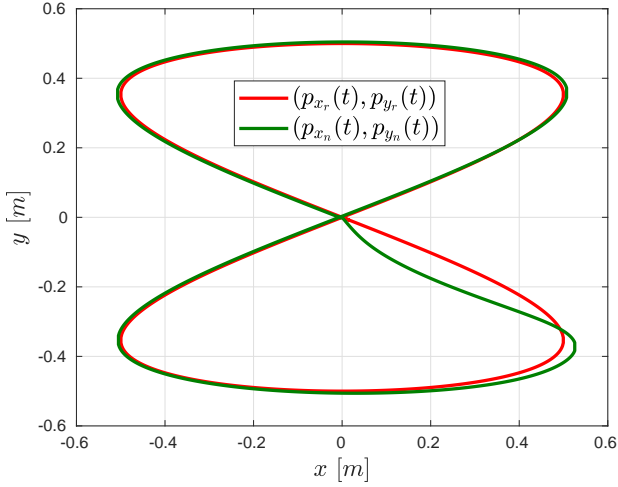}
	\caption{The desired reference trajectory $(p_{x_r}(t),p_{y_r}(t))$ (red), and the nominal system trajectory $(p_{x_\textup{n}}(t), p_{y_\textup{n}}(t))$ (green) under $u_\textup{n}$, for $t\in[0,8.5]$ seconds.}
	\label{fig:traj nominal}
\end{figure}
By also using $x(t) \in \widetilde{\mathcal{C}}_v$, for all $t\in[t_0,t']$ and the compactness of $\widetilde{\mathcal{C}}_v$, we conclude the boundedness of $x(t)$ and $\beta_v(h_v(x(t)))$, for all $t\in[t_0,t_{\max})$.  By further invoking  \cite[Th. 2.1.4]{bressan2007introduction}, we conclude that $t_{\max} = \infty$, and the boundedness of $x(t)$, $\beta(h(x_1(t)))$, $u(x(t))$ for all $t\in[t_0,\infty)$.
\end{proof}

\begin{figure}
	\centering
	\includegraphics[width=0.45\textwidth]{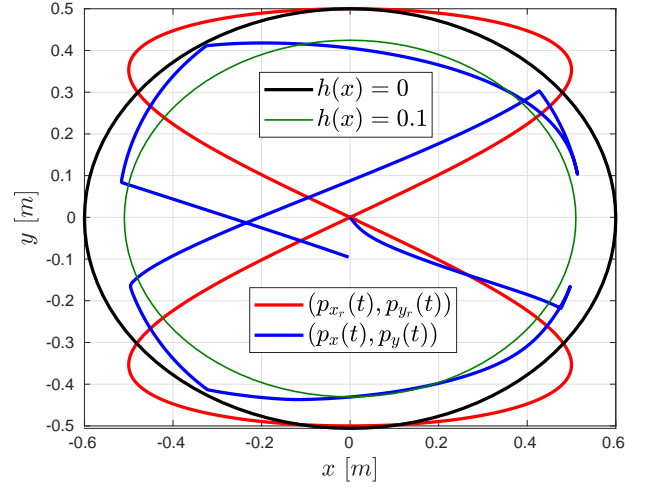}
	\caption{The desired reference trajectory $(p_{x_r}(t),p_{y_r}(t))$ (red), and the system trajectory $(p_{x}(t), p_{y}(t))$ (blue), for $t\in[0,8.5]$ seconds, along with the boundaries  $h(x)=0$ (black) and $h(x)=0.1$ (green).}
	\label{fig:traj safe}
\end{figure}

\section{Simulation Results} \label{sec:simulations}

We validate the proposed algorithm with a simulation example. More specifically, we consider an underactuated unmanned aerial vehicle (UAV) with state variables $x = [x_1,\dots,x_6] = [p_x,p_y,\phi,v_x,v_y,\omega]^\top$ evolving subject to the dynamics $\dot{p}_x = v_x$, $\dot{p}_y = v_y$, $\dot{\phi} = \omega$, and
\begin{align*}
    m\dot{v}_x &= -C_D^v v_x - u_1 \sin(\phi) - u_2 \sin(\phi) \\
    m\dot{v}_y &= -(mg + C_D^v v_y) + u_1 \cos(\phi) + u_2 \cos(\phi) \\
    2I\dot{\omega} &= - C_D^\phi \omega - l u_1 + l u_2,
\end{align*}
where $m=1.25$, $I=0.03$ are the quadrotor's mass and moment of inertia, respectively, $g=9.81$ is the gravity constant, $l=0.5$ is the arm length, and $C_D^v=0.25$, $C_D^\phi=0.02255$ are aerodynamic constants. We consider that the UAV aims to track the helicoidal trajectory $p_{x_r} \coloneqq \frac{1}{2}\sin(\frac{3}{2}t)$,  $p_{y_r} \coloneqq \frac{1}{2}\sin(\frac{3}{4}t)$ (see Fig. \ref{fig:traj nominal}) via an appropriately designed nominal control input $u_\textup{n}$. We wish to bound the position $(p_x,p_y)$ of the UAV 
through the sphere  $h(x) = 0.36 - \|[p_x,p_y]^\top\|^2$. We use the local safety controller \eqref{eq:control law general}, with $\beta = \frac{1}{h}$, $\mu_x = 0.1$,  $\kappa_x = 1$, $h_v(x) = 100 - e_2^2$, $\beta_v \coloneqq \frac{1}{h_v}$, $\mu_v = 100$, $\kappa_v = 100$, 
while setting  $\underline{\varepsilon} = 0.05$, $\bar{\varepsilon} = 5$ in Alg. \ref{alg:main alg}. The data measurement and hence the execution of Alg. \ref{algo:overapprox-datapoints} occurs every $0.1$ seconds. 
For the case when $\|\hat{g}^i(x_c)^\top \nabla{x_2}h_j(x_c)\| \leq \underline{\varepsilon}$ for some $x_c$, we use an optimization solver that aims to find an ellipsoidal $h_{j+1}(x)$ such that $\mathcal{C}_{j+1} \subset \mathcal{C}_j$ and maximize the value $\|\hat{g}^i(x_c)^\top \nabla{x_2}h_{j+1}(x_c)\|$. 

The simulation results from the initial condition $[0,0.2,0,-0.3,0,0]^\top$ are illustrated in Figs. \ref{fig:traj safe}-\ref{fig:G_error} for $t\in[0, 8.5]$ seconds. In particular, Fig. \ref{fig:traj safe} depicts the reference helicoidal trajectory (red) and the system trajectory under the safety controller \eqref{eq:control law general} (blue) along with the boundaries of the barrier $h(x) = 0$ (black) and local barrier function $h(x) = \mu = 0.1$ (green). 
One can verify that the system position is successfully confined in the set $\textup{Int}(\mathcal{C})$ defined by $h(x) > 0$, verifying thus the theoretical findings. Moreover, Fig. \ref{fig:error_2} depicts the evolution of the error $e_2(x(t)) = [e_{2_1}(x(t)),e_{2_2}(x(t))]$, which is successfully confined in the sphere imposed by $h_2(x)$; Fig. \ref{fig:rho_1} depicts the evolution of the discrete variable $\rho_1(t)$, from which it is concluded that $\|g^i(x)^\top \nabla_{x_2} h_v(x)\|$ falls below $ \underline{\varepsilon}$ several times and a new function $h_2(x)$ is found, as per Algorithm \ref{alg:main alg}. While $h_2$ is activated, however, $\|g^i(x)^\top \nabla_{x_2} h_2(x)\|$ is always above $\underline{\varepsilon}$, not requiring thus a new $h_3(x)$. Finally, Fig. \ref{fig:control inputs} depicts the required control input and Fig. \ref{fig:G_error} shows the evolution of the error norm $\|\widetilde{g}^i(x)\|$, $i\in\mathbb{N}$. From Fig. \ref{fig:G_error} it can be verified  that the condition \eqref{eq:local cond2} does not always hold; nevertheless, the proposed control algorithm is still able to guarantee the system safety. Hence, on can conclude that it is not necessary for Theorems \ref{th:local}, \ref{th:final}.

\begin{figure}
	\centering
	\includegraphics[width=0.45\textwidth]{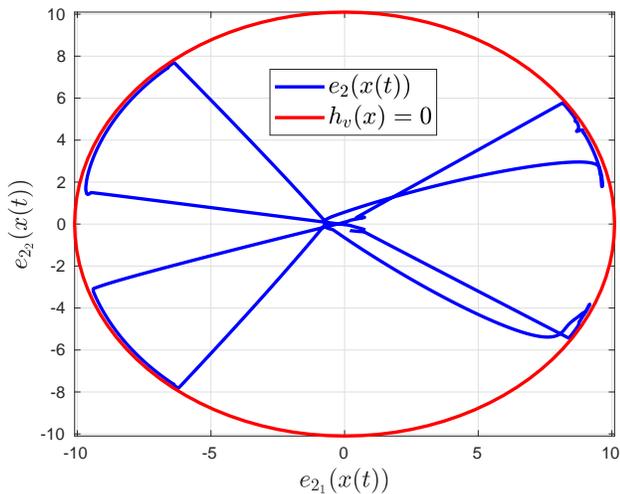}
	\caption{The velocity error signal $e_2(x(t))$ (blue) for $t\in[0,8.5]$ seconds, and the boundary $h_v(x)=0$ (red).}
	\label{fig:error_2}
\end{figure}

\begin{figure}
	\centering
	\includegraphics[width=0.45\textwidth]{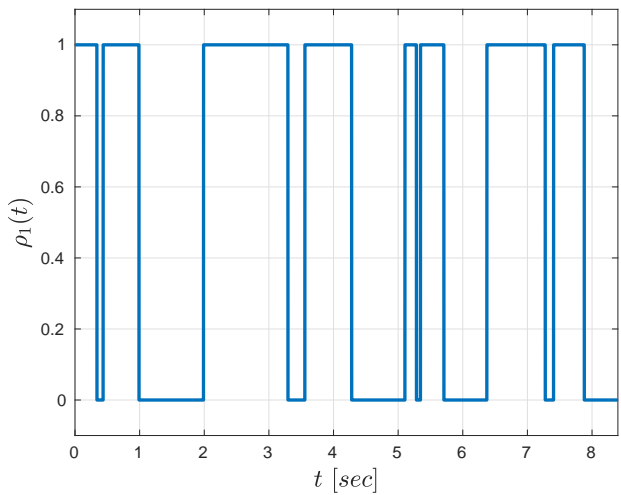}
	\caption{The evolution of the discrete signal $\rho_1(t) \in \{0,1\}$ for $t\in[0,8.5]$ seconds.}
	\label{fig:rho_1}
\end{figure}

\begin{figure}
	\centering
	\includegraphics[width=0.4\textwidth]{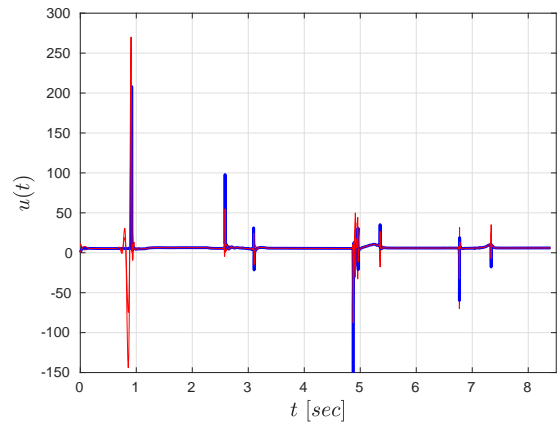}
	\caption{The evolution of the control inputs $u_1(t)$ (red) and $u_2(t)$ (blue) for $t\in[0,8.5]$ seconds. }
	\label{fig:control inputs}
\end{figure}

\begin{figure}
	\centering
	\includegraphics[width=0.4\textwidth]{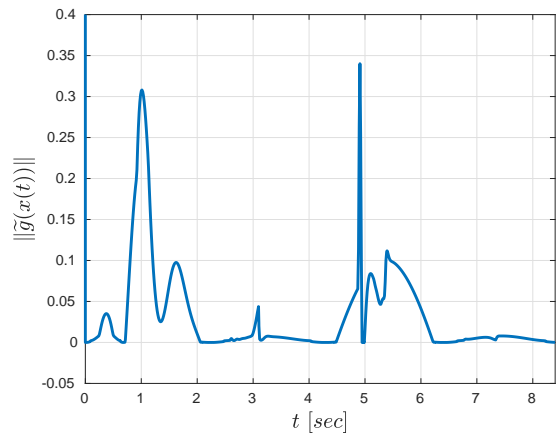}	
	\caption{The evolution of the error $\widetilde{g}^i(x(t))$ for $t\in[0,8.5]$ seconds. }
	\label{fig:G_error}
\end{figure}

 \section{Conclusion and Future Work} \label{sec:conclusion}
We consider the safety problem for a class of $2$nd-order nonlinear unknown systems. We propose a two-layered control solution, integrating approximation of dynamics from limited data with closed form nonlinear control laws using reciprocal barriers. Future efforts will be devoted towards extending the proposed framework to stabilization/tracking.

\bibliographystyle{IEEEtran}
\bibliography{references_journal}

\end{document}